%% file: despostfull.tex
\documentclass[letterpaper]{article}
\usepackage{tikz}
\usetikzlibrary{arrows,chains,,positioning,scopes}

\usepackage{aaai20arxiv}  
\usepackage{times}  
\usepackage{helvet} 
\usepackage{chngcntr}
\usepackage{courier}  
\usepackage[hyphens]{url}  
\usepackage{graphicx} 
\usepackage{graphicx}  
\usepackage{paralist}
\usepackage{multicol}
\usepackage[utf8x]{inputenc}
\usepackage{fullpage}
\usepackage{xspace}
\usepackage{stmaryrd}
\usepackage{etoolbox}
\usepackage{amsmath,amssymb}
\usepackage{amsthm}
\usepackage{algorithmicx}
\usepackage{algpseudocode}

\usepackage{paralist}
\input{macros}

\title{Balancing expressiveness and inexpressiveness in  view design}

\author{Michael Benedikt and Pierre Bourhis and Louis Jachiet and Efthymia Tsamoura}

\begin{document}

\maketitle
\input{abstract}

\input{intro}

\input{contrib}

\input{organ}

\input{prelims}

\input{noconstraints}

\input{constraints}

\input{conc}

\bibliographystyle{aaai}
\bibliography{refs}
\newpage
\onecolumn
\appendix
\input{appendix}

\end{document}

%% file: macros.tex
\newcommand{\ph}{\kw{pair} \kw{height}}
\newcommand{\D}{\mathcal{D}}
\newcommand{\I}{\mathcal{I}}
\newcommand{\E}{\mathcal{E}}

\newcommand{\Vcal}{\views}
\newcommand{\Pcal}{{\cal P}}
\newcommand{\Scal}{{\cal S}}

\renewcommand{\vec}[1]{\mathbf{#1}}
\newcommand{\tocheck}{}

\newcommand{\N}{\ensuremath{\mathcal{N}}\xspace}
\newcommand{\safe}{\kw{safe}}
\newcommand{\dvars}{\kw{Vars}}
\newcommand{\stretchf}{\kw{Str}}

\newcommand{\globalinst}{{\mathcal D}}
\newcommand{\globinst}{\globalinst}

\newcommand{\<}{\langle}
\renewcommand{\>}{\rangle}
\newcommand{\critelement}{*}
\newcommand{\criticalelement}{\critelement}
\newcommand{\critbinding}{\critelement \ldots \critelement}

\newcommand{\sources}{\kw{Srcs}}
\newcommand{\globview}{{\mathcal V}}
\newcommand{\aschema}{{\mathcal S}}

\newcommand{\utility}{Q} 

\newcommand{\minq}{\kw{min}}

\newcommand{\michael}[1]{}

\newcommand{\determinacy}{\kw{Determinacy}}
\newcommand{\defeq}{:=}
\newcommand{\chase}{\kw{Chase}}

\newtheorem{example}{Example}

\newcommand{\unprime}{\kw{UnPrime}}

\newcommand{\true}{\textbf{true}}

\newcommand{\conp}{\textsc{CoNP}}

\newcommand{\sigmatwop}{\Sigma_2^p}

\newcommand{\V}{{\mathcal V}}
\newcommand{\Varb}{{\mathcal V}}
\newcommand{\varb}{\Varb}

\newcommand{\views}{{\mathcal V}}
\newcommand{\view}{\kw{V}}

\newcommand{\util}{Q}

\newcommand{\treatment}{\kw{Trtmnt}}

\newcommand{\patient}{\kw{Patient}}
\newcommand{\patientid}{\kw{pid}}
\newcommand{\treatmentinfo}{\kw{tinfo}}
\newcommand{\age}{\kw{age}}
\newcommand{\myaddress}{\kw{address}}
\newcommand{\treatdate}{\kw{tdate}}

\newcommand{\canondb}{\kw{canondb}}
\newcommand{\adom}{\kw{adom}}
\newcommand{\instance}{{\mathcal I}}
\newcommand{\inst}{\instance}
\newcommand{\forwview}{\kw{ForView}}
\newcommand{\backview}{\kw{BackView}}

\newcommand{\canonv}{\kw{CanView}}
\newcommand{\canonview}{\canonv}

\newcommand{\asecretquery}{p}
\newcommand{\secret}{\asecretquery}

\newcommand{\equivr}{\equiv}

\newcommand{\shuf}{\mu}
\newcommand{\canoncont}{\kw{CanCtxt}}

\newcommand{\sjvars}{\kw{SJVars}}
\newcommand{\svars}{\kw{SVars}}

\newcommand{\kw}[1]{{\mathsf{#1}}\xspace}

\newtheorem{theorem}{Theorem}
\newtheorem{claim}{Claim}

\newtheorem{proposition}{Proposition}
\newtheorem{corollary}{Corollary}

  \newtheorem{definition}{Definition}

\renewcommand{\phi}{\varphi}

\newtheorem{lemma}{Lemma}

\newcommand{\myparagraph}[1]{{\bf #1.}}
\newcommand{\myeat}[1]{}

\usepackage{xcolor}

%% file: abstract.tex
\begin{abstract}
We study the design of data publishing mechanisms
that allow a collection of autonomous distributed datasources to collaborate to support queries.
A common mechanism  for data publishing is via \emph{views}:  functions that expose
derived data to users, usually specified as declarative queries. Our autonomy assumption is
that the views must be on individual sources, but with the intention of supporting integrated
queries.
In deciding what data to expose to users, two considerations must be balanced. The views must be sufficiently expressive
to support queries that users want to ask -- the \emph{utility} of the publishing mechanism.
But there may also be some expressiveness restriction. Here we consider two restrictions,
a \emph{minimal information} requirement, saying that the views should reveal as little as possible while supporting the utility
query, and a \emph{non-disclosure} requirement, formalizing the need  to prevent external users from computing  information that data owners do not want revealed.
 We investigate the problem of designing views that satisfy both an expressiveness and an inexpressiveness requirement,
for views in a restricted declarative language
(conjunctive queries), and for arbitrary views.
\end{abstract}

%% file: intro.tex
\section{Introduction} \label{sec:intro}
The value of data is increased when data owners
make their data  available through publicly-accessible
interfaces. The value is magnified even further
when multiple  data owners publish information from related
datasets; this allows users to answer queries that require
linking information across datasources.

But the benefits  of data publishing come with a corresponding
risk of revealing too much.  For example, there may be information
that a data owner wishes to protect, and a user
may be able to infer this information either
from the published data in isolation, or from the
data published by all parties as a whole.
There is thus a need to provide data publishing mechanisms
that are simultaneously expressive enough to be useful -- they enable users
to answer appropriate queries -- while satisfying some expressiveness restriction.

In data publishing much of the  focus
has been on disclosure via the familiar  mechanism of
\emph{views} -- declarative queries whose
output is made available to users as a table. 
In this context the competing requirements on
a publishing mechanism have been primarily studied in isolation.
There is  extensive work on analysis
of the utility of a set of views: namely given a query, can it be answered
using the views, see, e.g. \cite{DBLP:journals/vldb/Halevy01,DBLP:journals/jcss/CalvaneseGLR12}.
There has also been research concerning analysis
of whether a given set of views 
obeys some expressiveness \emph{restriction}. The negation of answerability
is clearly too weak a restriction, since it just
guarantees that on some instance the query answer can not be computed from the view
images. 
A relevant query-based notion of expressiveness  restriction is \emph{data-independent privacy}:
given the views and a set of  ``secret'' queries, check that the secret query answers
can not be computed from the views on any source data. Variants of this problem have been studied in
\cite{dnprivacy,lics16,privacyjair,ijcai19}. Expressiveness restrictions with a similar
flavor have also been studied in the context of ontologies \cite{DBLP:conf/semweb/BonattiS13}.
But the question of whether there is an expressiveness restriction 
for views that does not require the specification of a particular set of secrets,
as well as the question of how one obtains views that satisfy both
expressiveness and inexpressiveness requirements,
has not been considered 
in the context of traditional queries and views, to the best of our knowledge.

\myeat{
In the multi-agent systems community there is 
work concerned with  protocols allowing a set of agents to learn a distribution
on source data
without revealing certain local information (e.g. \cite{sadi}). The notion of security is ``exact and
information-theoretic'', as in this work:  they deal
with getting the exact answer to a query using a deterministic algorithm, but without restricting the computational power used in answering. 
But the setting, unlike in this work,  is propositional.
}

A larger body of work comes from privacy research, considering the design
of mechanisms  achieving a mix of expressiveness (``utility'') and inexpressiveness
(``privacy'') goals. But the focus
is  on 
probabilistic transformations or, more generally, probabilistic
protocols (see e.g. \cite{infotheoreticprivccd,diffpriv,diffprivj}). The guarantees are probabilistic,
 sometimes alternatively or additionally
with computational restrictions on an external party.
\myeat{
Probabilistic protocols iterate rounds of messages,
either broadcast to all owners or point-to-point between data owners. The behavior of the protocol
includes randomization, which allows the mechanisms
to achieve very strong privacy guarantees. In particular
they can achieve guarantees that do not require the specification of a privacy goal in terms of a specific set of 
protected queries. 
For example, in  information-theoretic privacy (e.g. \cite{infotheoreticprivccd})  the goal
is a mechanism that can guarantee  that any data owner ``learns nothing'' (in probabilistic terms) from
a run of the protocol that they cannot learn from the answer to the utility
query in isolation. Research
 in \emph{differential privacy} \cite{diffpriv,diffprivj} has introduced mechanisms
that  guarantee that the probability of an output will not be impacted (up to
some $\epsilon$ that is a parameter to the analysis) by the presence or absence
of any tuple in the database.  Roughly speaking, the attacker will learn nothing
about any ground fact in the database.
}
Recent efforts \cite{pricingprivate} have considered a family of mechanisms that 
look at database queries, with the utility of a mechanism defined (as in our case)
using the notion of query determinacy. But randomness still plays a central
role in the mechanism and in the definition of privacy.

In contrast, we consider the question of designing views that use \emph{traditional database queries},
with no randomization,  so that
 conflicting requirements of expressiveness and inexpressiveness are satisfied. Both our expressiveness
and inexpressiveness requirements will be given in terms of \emph{exact information-theoretic criteria}:
they will be defined in terms of what queries can be answered exactly (as opposed
to probabilistically)  by a party with unlimited computation power. Due both to the difference in the mechanisms
we consider and the requirements we impose,
our contribution has a very different flavor  from  prior lines of work.
It also differs from work on secure querying over distributed data \cite{jenniesecurequery}.
There the goal is to support ad hoc querying of traditional database queries in the presence
of privacy restrictions; but they allow the use of encryption as a query language primitive.

\input{runningexample}

Our goal here is to look at the problem of designing 
independent views over multiple relational datasources that satisfy both expressiveness and
inexpressiveness requirements.
Our expressiveness requirement (usefulness) will be phrased in terms of the ability to answer  a 
relational query, where answering is the traditional deterministic notion used in database theory and knowledge representation.
For expressiveness limitations we require the views to be
useful but to \emph{minimize information} within a class of views.
We also consider an  expressiveness limitation specified by  \emph{non-disclosure} of 
a  set of \emph{secret queries}. Our contributions include formalizing these notions,
characterizing when minimal information views exist and what form they take, and determining
when views exist that satisfy utility and expressiveness limitations. We look at these problems
both for views given in the standard view language of conjunctive queries, and for arbitrary views.
We also consider the impact of background knowledge on these problems.

%% file: runningexample.tex
\begin{example} \label{ex:run}
A health study hosted by a government agency holds information
about certain treatments, with an abstraction of the data being a database
with schema $\treatment(\patientid, \treatmentinfo, \treatdate)$, where $\patientid$
might be a national insurance number.

Demographic information about patients is stored by another agency,
in a table $\patient(\patientid, \age, \myaddress)$.
The agencies are completely autonomous, perhaps even in distinct administrative regions.
But  they want to co-operate to support certain queries over the data that legitimate
researchers might wish; for example about the relationships between treatment and age:
\begin{flalign*}
Q(\treatmentinfo, \age)  = \\
\exists \patientid, \myaddress, \treatdate.
\treatment(\patientid, \treatmentinfo, \treatdate) \\
\wedge \patient(\patientid, \age, \myaddress)
\end{flalign*}

Of course, the parties could agree
to an encryption schema on the patient identifiers, and then
expose encrypted versions of their local schemas.  But
this would require both strong co-operation of the parties,
and the use of views beyond traditional database queries.


But assuming that the parties are restricted to using traditional queries,
what is the most restrictive thing that they can do while supporting the ability to answer $Q$?
Intuitively, the most restrictive views would correspond to  one party revealing
 the projection of the $\treatment$ table on $\patientid$ and $\treatmentinfo$ and the other
revealing the projection of the $\patient$ table on $\patientid$
and $\age$.

If this intuition is correct, it would imply that nothing the parties can do with
traditional queries can avoid  revealing which patients had particular treatments. That is,
they have no choice but to allow an external party to learn the answer to query
\[
\secret = \exists \treatdate. \treatment(\patientid, \treatmentinfo, \treatdate)
\]

Our results will validate this obvious answer -- in this example, the projections described above are the
CQ views that reveal minimal information, and one can not support the disclosure of the join query while
protecting a query on these join attributes.
 Further we will show that even using encryption --
indeed, even using any deterministic function -- the parties can not not reveal less information  while
supporting exact answering of the query $Q$.
In contrast,  we will also show that in some cases the parties
can obtain combinations of expressiveness  and inexpressiveness requirements
by using counter-intuitive view combinations.
\end{example}

%% file: contrib.tex
\myparagraph{Contributions}
We will not be able to present a full picture of the view design problem in this
work. We will deal only with the case of utility and secret queries given as conjunctive queries (CQs),
the analogs of  SQL basic SELECTs, and deal only with 
restricted source integrity constraints.  But we believe our results suffice to  show that our  formulation captures an important
trade-off  in schema design, and that query language
expressiveness issues come to the fore. In technical terms, we make the following contributions:

\begin{compactitem}
\item We formalize the idea of balancing expressiveness 
and inexpressiveness in distributed views, via the notions of useful distributed views 
as well as   ``minimal information'' requirements among these views.
\item We show that there are useful views with minimal information among the
set of  CQ views, and also among the set of all views, but that these may differ.
\item In contrast to the above, we show that to obtain useful views with minimal information
 we do not need to go that far beyond CQ views: it suffices
to use views defined either
using  unsafe disjunction of CQ views, or in relational algebra.
\item We show that the above results extend to the presence
of background knowledge that is local to each source.
\item We examine the impact of background knowledge that relates multiple sources,
looking at the simplest kind of relationship, the replication of a relation across
sources. We show that we may no longer have useful views with minimal information,
but we can get useful views that are minimal in terms of the secrets they disclose.
In the process we show that replication can be exploited to allow
a view designer to reveal certain queries while not disclosing others.
\end{compactitem}

A diverse set of techniques are introduced to study these problems.
For investigating the  use of CQ views to satisfy both expressiveness
and inexpressiveness restrictions,
 we employ an analysis of the chase proofs that witness  determinacy.
For studying the use of arbitrary views, both in the absence of background knowledge
and in the presence of only local background knowledge, we use
a Myhill-Nerode  style characterization of
when two local sources are interchangeable in terms of their impact on a utility CQ, and
then we relate
this characterization to certain partial symmetries of the utility queries (``shuffles'').
Our analysis of the impact of replication constraints relies on a product
construction, which is the key to allowing us to generate useful views
that  disclose the minimal number of secrets.

%% file: organ.tex
\myparagraph{Organization}
Section \ref{sec:prelims} gives database and logic preliminaries,
and then goes on to formalize our expressiveness and inexpressiveness requirements.
Section \ref{sec:cqviews} deals with the variant of the problem where the only
views considered are conjunctive query views, while Section \ref{sec:arbviews}
shows how the situation changes when arbitrary views can be utilized.
Section \ref{sec:localconstraints}   contains extensions  when 
background knowledge local to a source is present, while
Section \ref{sec:replication} consider background knowledge about connections
across sources.
We close with conclusions in Section \ref{sec:conc}.
Full proofs are deferred to the appendix.

%% file: prelims.tex
\section{Preliminaries} \label{sec:prelims}

\input{basicdefs}

\input{privacydefs}

\input{tools}

%% file: basicdefs.tex
\subsection{Basic definitions}
The bulk of this subsection reviews standard definitions from databases and knowledge
representation.
But it includes two notions,  DCQs and distributed schemas, that
are less standard.

\myparagraph{Databases and queries}
A \emph{schema} consists of a finite set of relations, each with an associated
arity  (a non-negative number).
An \emph{instance} of a schema is an assignment of each
relation in the schema of arity $n$ to a collection (finite or infinite) of $n$-tuples of elements.
Given an instance $\inst$ and a relation $R$, we let $\inst(R)$ be the $n$-tuples
assigned to $R$ in $\inst$.
The \emph{active domain} of an instance $\inst$, denoted $\adom(\inst)$, is the set of elements occurring
in some tuple of some $\inst(R)$. 
A \emph{fact} consists of a relation $R$ of arity $n$
and an $n$-tuple $\vec t$.  We write such a fact as $R(\vec t)$.
An instance can equivalently be thought of as a set of facts.

An $n$-ary \emph{query} is a function from instances of a fixed schema $\aschema$ to some set.
We refer to $\aschema$ as the \emph{input
schema} of  the query.
A \emph{Conjunctive Query} (CQ) is a formula of the form
${\exists \vec y ~ \bigwedge_i A_i}$
where $A_i$ are atoms 
from the schema.
A Boolean CQ (BCQ) is a CQ with no free variables.  
Following  the notation used in several places in the literature (e.g. \cite{arenasucq}) we define
a \emph{union of CQs} (UCQ) to be a disjunction of CQs satisfying
the \emph{safety condition} where the free variables of each disjunct are the same.
Disjunctions of CQs
where the safety condition is dropped will play a key role in this paper.
The terminology for such queries is less standardized, but we refer to them as \emph{disjunctions of CQs} (DCQs).
UCQs can be extended to
\emph{relational algebra}, the standard algebraic presentation of first-order relational queries:
queries are build up from relation symbols by union, difference, selection, and product \cite{AHV}.

For a logical formula $\rho$ with free variables ${x_1, \dots, x_n}$ and an instance $\inst$,
a \textit{variable binding} $\sigma$ for $\rho$ in $\inst$ is a mapping taking each $x_i$ to an element of $\adom(\inst)$. 
We can apply
$\sigma$ to $\rho$ to get a new formula $\sigma(\rho)$ where each $x$ is replaced
by $\sigma(x)$. Assuming an ordering of the free variables as ${x_ 1, \dots, x_n}$ we may identify
a $k$-tuple $\vec t$,
writing $\rho(\vec t)$ to mean that $t_i$ is substituted for $x_i$ in $\rho$.

A \emph{homomorphism} between instances $\inst_1$ and $\inst_2$ is a function $f$
from $\adom(\inst_1)$ to $\adom(\inst_2)$ such that ${R(c_1, \dots, c_m) \in \inst_1}$ implies
${R(f(c_1), \dots, f(c_m)) \in \inst_2}$.
The notion of homomorphism from a CQ to an instance and
from a CQ to a CQ is defined similarly.
In the case of
CQ-to-CQ homomorphisms we additionally require that the mapping be the identity on any free variables or constants.
The \emph{output} of a CQ  $Q$ on an instance $\inst$, denoted $Q(\inst)$ consists of  the restrictions
to free variables of $Q$ of the homomorphisms of $Q$ to $\inst$. The output of a UCQ is defined
similarly.
We can choose an ordering of the free variables of  $Q$,
and can then say that the output of $Q$ on $\inst$ consists of $n$-tuples. We write
$\inst, \vec t  \models Q$ 
for an $n$-tuple $\vec t$, if ${\vec{t} \in Q(\inst)}$. 
We analogously define ${\inst, \sigma \models Q}$ for a variable binding $\sigma$. 
\myeat{A CQ with
$n$ free variables thus defines an $n$-ary query, and similarly for
a UCQ, and more generally any logical formula with $n$ free variables.}
We sometimes refer to a homomorphism of a BCQ into an instance
as a \emph{match}.
For a logical formula $\rho(\vec x)$ and a tuple of elements $\vec t$,  $\rho(\vec t)$
denotes the formula where each $x_i$ is substituted with $t_i$.

A CQ $Q_0$ is a \emph{subquery} of a CQ $Q$ if the atoms of $Q_0$ are a subset
of the atoms of $Q$ and a variable of $Q_0$ is free in $Q_0$ if and only if it is free in $Q$.
A \emph{strict subquery} of $Q$ is a subquery of $Q$ that is not $Q$ itself; $Q$ is \emph{minimal} if there is no homomorphism from $Q$ to a strict subquery of $Q$.

A \emph{view} over a schema $\aschema$ consists of an $n$-ary relation $\view$ and a corresponding $n$-ary
query $Q_\view$ over relations from $\aschema$. Given a collection of views $\views$ and instance
$\inst$,  the \emph{view-image} of $\inst$, denoted $\views(\inst)$ is the instance
that interprets each $\view \in \views$ by $Q_\view(\inst)$.
We thus talk about \emph{CQ views},
\emph{UCQ views}, etc: views defined by formulas within a class.

\myparagraph{Distributed data and views}
A  \emph{distributed schema} (d-schema) $\aschema$ consists of
a finite set of \emph{sources} $\sources$, with
each source $s$  associated with a \emph{local schema}
$\aschema_s$. We assume that the relations in distinct  local schemas are pairwise disjoint.
In Example \ref{ex:run} our distributed schema consisted of two sources, one containing
$\treatment$ and the other containing $\patient$.
A \emph{distributed instance} (d-instance) is an instance of a distributed schema.
For a source $s$, an $s$-instance is an instance of the local schema $\aschema_s$.
Given a d-instance $\globinst$, we denote by $\globinst_s$ the restriction of $\globinst$ to relations
in $s$. 
If d-schema $\aschema$ is the disjoint union of $\aschema_1$ and $\aschema_2$,
and we have  sources $\inst_1$ for $\aschema_1$ and $\inst_2$ for $\aschema_2$, then
we use $(\inst_1, \inst_2)$ to denote the union of $\inst_1$ and $\inst_2$, which is
an instance of $\aschema$.

For a given d-schema a \emph{distributed view} (d-view) $\globview$ is an assignment
to each source $s$ of a finite set $\globview_s$ of views over its local schema.
Note that here is our ``autonomy'' assumption on the instances: we are free to design
views on each local source, but views can not cross sources.
We can similarly talk about CQ-based d-views, relational algebra-based d-views, etc.

\myparagraph{Existential Rules} Many semantic relationships between
relations can be described using \emph{existential rules}, with this paper focusing
on the variant of existential rules that  are
also called \emph{Tuple Generating Dependencies}. These are
logical
sentences of the form $\forall \vec x. \lambda  \rightarrow
  \exists \vec y. \rho$, where $\lambda$ and
$\rho$ are conjunctions of relational atoms.
The notion of a formula $\rho$
\emph{holding} in $\inst$ (or $\inst$ \emph{satisfying} $\rho$, written ${\inst \models
\rho}$) is the standard one in first-order logic. 
A \emph{trigger} for $\tau$ in $\inst$ is a homomorphism
$h$ of $\lambda(\vec x)$ into $\inst$. Moreover, a trigger $h$ for $\tau$ is  \emph{active}
if no extension of $h$ to a homomorphism of $\rho(\vec x,\vec y)$ into $\inst$
exists.
Note that a dependency $\tau$ is satisfied in $\inst$ if there does not exist an active trigger
for $\tau$ in $\inst$.

\myeat{
Let $\Sigma$ be a set of rules, $\inst$ be a finite
instance, and $Q$ be a BCQ.  We write ${\inst \wedge \Sigma \models Q}$ to mean that
every instance containing $\inst$ and satisfying $\Sigma$ also satisfies $Q$.
For two CQs $Q$ and $Q'$, we similarly write ${Q \wedge \Sigma \models Q'}$ to mean
that every instance satisfying $Q$ and $\Sigma$ satisfies $Q'$.
}
\input{chase}

%% file: chase.tex
\myparagraph{The chase}
The results in Section \ref{sec:cqviews} will make use
of the characterization of logical entailment between CQs in the presence
of rules in terms of the chase procedure \cite{chaseorig,fagindataex} which we 
review here.
The chase modifies an instance by a sequence of \emph{chase steps} until all
dependencies are satisfied. Let $\inst$ be an instance,  and consider a rule
 $\tau=\forall \vec x. \lambda  \rightarrow \exists \vec y. \rho$.
 Let  $h$ be a
trigger for $\tau$ in $\inst$. Performing a chase step for
$\tau$ and $h$ to $\inst$ extends $\inst$ with each facts of the conjunction ${h'(\rho(\vec
x,\vec y))}$, where $h'$ is a substitution such that ${h'(x_i) = h(x_i)}$ for
each variable ${x_i \in \vec x}$, and $h'(y_j)$, for each ${y_j \in \vec y}$,
is a fresh labeled null that does not occur in $\inst$. 

For $\Sigma$ a set of existential rules  and $\inst$ an instance a
\emph{chase sequence} for $\Sigma$ and $\inst$ is a (possibly infinite) sequence
${\inst_0, \inst_1, \ldots}$ such that ${\inst = \inst_0}$ and, for each ${i > 0}$, instance
$\inst_i$ if it exists
is obtained from $\inst_{i-1}$ by applying a successful chase
step to a dependency ${\tau \in \Sigma}$ and an active trigger $h$ for $\tau$
in $\inst_{i-1}$. The sequence must be \emph{fair}: for each ${\tau \in \Sigma}$,
each ${i \geq 0}$, and each active trigger $h$ for $\tau$ in $\inst_i$, some ${j >
i}$ must exist such that $h$ is not an active trigger for $\tau$ in $\inst_j$
The \emph{result} of a
chase sequence is the (possibly infinite) instance $\inst_\infty = \bigcup_{i \geq
0}  \inst_i$. 
We use $\chase_\Sigma(\inst)$
to denote the result of any chase sequence for $\Sigma$ on $\inst$.

A finite chase sequence is \emph{terminating}. A set of dependencies $\Sigma$
\emph{has terminating chase} if, for each finite, instance $\inst$, each
chase sequence for $\Sigma$ and $\inst$ is terminating. For such $\Sigma$, the
chase provides an effective approach to testing if  $\inst \wedge \Sigma \models Q$:
we compute (any) chase $\inst_\infty$ for $\Sigma$ and $\inst$ and check if $Q$ holds \cite{fagindataex}.
We can similarly test if $Q \wedge \Sigma \models Q'$ by chasing $\canondb(Q)$ with $\Sigma$
and then checking $Q'$.
Checking if a set of dependencies $\Sigma$ has terminating chase
is undecidable \cite{chaserevisited}. \emph{Weak acyclicity} \cite{fagindataex}
was the first sufficient polynomial-time condition for checking if $\Sigma$ has
terminating chase. Stronger sufficient (not necessarily polynomial-time)
conditions have been proposed subsequently \cite{acyclicity,onet}.

%% file: privacydefs.tex
\subsection{Problem formalization}
We now give the key  definitions in the paper, capturing our expressiveness
requirements (``usefulness'')  and expressiveness limitations
(``minimally informative'' and ``non-disclosing'').
Our expressiveness requirement is via the notion of determinacy \cite{NSV}, formalizing
the idea that on any instance
there is sufficient information in the views to recapture the query.

\begin{definition}
Two d-instances $\globinst$ and $\globinst'$ are \emph{indistinguishable}
by a d-view $\views$ (or just $\views$-indistinguishable)
if ${V(\globinst)=V(\globinst')}$ holds, for each view ${V \in \views}$.
\end{definition}
Since each view $V \in \globview_s$ is defined over relations occurring only in $\globview_s$,
we can equivalently say that $\globinst$ and $\globinst'$ are $\views$-indistinguishable if  
${V(\globinst_s)=V(\globinst'_s) }$ holds, for each ${V \in \views_s}$. 

\begin{definition}
A d-view $\views$ \emph{determines} a query $Q$ 
\emph{at} a d-instance $\globinst$ if ${Q(\globinst')=Q(\globinst)}$ holds, 
for each  $\globinst'$ that is $\views$-indistinguishable from $\globinst$. 

The d-view $\globview$ is \emph{useful} for $Q$ if 
$\globview$ determines $Q$ on every d-instance (for short, just ``$\globview$ determines $Q$'').
\end{definition}
Usefulness for a given query $Q$ will be our expressiveness requirement on d-views. Our first inexpressiveness requirement captures the idea that we want to reveal as little as possible:

\begin{definition}[Minimally informative useful views]
Given a class of views $\mathcal C$ and a query $Q$, we say that
a d-view $\views$ is  \emph{minimally informative useful d-view for $Q$ within $\mathcal C$}
(``Min.Inf. d-view'')
if $\views$ is useful for $Q$ and for any other d-view $\views'$ useful
for $Q$ based on views in $\mathcal C$,
$\views'$ determines the view definition of each view in $\views$.
\end{definition}

We look at another inexpressiveness requirement that requires an external party to not learn
about another query.
\begin{definition}[Non-disclosure]
A \emph{non-disclosure function} specifies, for each query $\secret$, the  set of
d-views $\views$ that are said to disclose $\secret$. We require such a function $F$ to be \emph{determinacy-compatible}:
if $\views_2$ discloses $\secret$ according to $F$ and $\views_1$ determines each view in $\views_2$, then
$\views_1$ also discloses $\secret$ according to $F$.  If $\views$ does not disclose
$\secret$, we say that $\views$  is non-disclosing for $\secret$ (relative to the given
non-disclosure function).
\end{definition}
When we can find minimally informative useful d-views, this tells us something about non-disclosure,  since
it is easy to see that if we are looking to design views that
are useful and non-disclosing, it suffices
to consider  minimally informative views, assuming they exist:
\begin{proposition} \label{prop:minprivate}
Suppose $\views$ is a minimally informative useful d-view for $\util$ within $\mathcal C$, and
there is a d-view based on views in $\mathcal C$ that is useful for $\util$ and non-disclosing for $\secret$ according
to non-disclosure function $F$. Then $\views$ is useful for $\util$ and non-disclosing for $\secret$ according
to $F$.
\end{proposition}
There are many disclosure functions that are determinacy-compatible.
But in our examples, our complexity results,
and in Section \ref{sec:replication}, we will focus  on a specific non-disclosure function, whose intuition
is that an external party ``never infers any answers''.

\begin{definition}
A d-view $\globview$ is
\emph{universal non-inference non-disclosing} (UN non-disclosing) for a CQ $\secret$  if for each instance $\inst$
and each tuple $\vec t$ with $\inst,  \vec t \models  \secret$,
 $\globview$ does not determine $\secret(\vec t)$ at $\inst$.
Otherwise $\globview$ is said to be \emph{UN disclosing} for $\secret$.
\end{definition}
This non-disclosure function is clearly determinacy-compatible, so Proposition \ref{prop:minprivate} will apply to it. Thus we will be able to utilize UN non-disclosure as a means of showing
that certain d-views are \emph{not} minimally informative.

\myparagraph{Variations:  background knowledge, and finite instances}
All of these notions can be additionally parameterized by background knowledge
$\Sigma$, consisting of integrity constraints in some logic.
Given d-instance $\globinst$ satisfying $\Sigma$ and
a d-view $\globview$, $\globview$ \emph{determines} a query $Q$ over the d-schema
at $\globinst$ relative to $\Sigma$ if:
for every  $\globinst'$ \emph{satisfying $\Sigma$} that is $\globview$-indistinguishable from $\globinst$,
$Q(\globinst')=Q(\globinst)$.
We say that $\globview$ is \emph{useful} for a query $Q$ relative to $\Sigma$ if
it determines $Q$ on every d-instance \emph{satisfying $\Sigma$}.
We say $\globview$  is  
UN non-disclosing for query $\secret$ with
respect to $\Sigma$ if
 $\globview$ does not determine $\secret$ \emph{on any d-instance satisfying $\Sigma$}.

By default, when we say ``every instance'', we mean all instances, finite or infinite.
There are variations of this problem requiring the quantification in both non-disclosure
and utility to be over finite instances.

\myparagraph{Main problem}
We  focus on the problem  of determining whether minimally informative useful d-views exist for a given query $\util$ and class $\mathcal C$, and characterizing
such views when they do exist.
When minimally informative useful d-views do not exist, we  consider the problem of obtaining
a d-view that is useful for $\util$ and which minimizes the set of secrets $\secret$ 
that are UN disclosed for
$\secret$. 
We refer to $\util$ as the
\emph{utility query}, and $\secret$ as the \emph{secret query}.

\myparagraph{Discussion}
Our notion of utility of views is \emph{information-theoretic and exact}: a view is useful if  a party
with access to the view can compute the exact output of the query (as opposed to the correct
output with high probability),  with no limit on how difficult the computation may be.
 The generality of this notion  will make our negative results stronger. And 
it turns out the generality will not limit
our positive results, since these will be realized by very simply views.

Our notion of minimally informative is likewise natural if one seeks an ordering  on sets of views
measuring the  ability to support exact information-theoretic query answering.
Our query-based inexpressiveness notion, non-disclosure, gives a way of seeing the impact of minimally informative useful views
on protecting information, and
 it is also based on
information-theoretic and  exact notions.
We exemplify our notion of non-disclosure function with UN non-disclosure, which has been studied in  prior work 
under several different names \cite{dnprivacy,lics16,privacyjair,ijcai19}.
 We choose the name ``non-disclosing'' rather than ``private'' for all our query-based expressiveness
restrictions,
since they are clearly very  different
from  more traditional probabilistic privacy guarantees \cite{diffprivj}.
On the one hand the UN non-disclosure guarantee is weak in that $\secret$ is considered safe for $\views$ (UN non-disclosed) if an attacker
can never infer that $\secret$ is true with absolutely certainty. Given a distribution on source instances,
the information in the views may still increase the likelihood that $\secret$ holds.
On the other hand, the notion is quite strong in that it must hold on \emph{every} source instance.
Thus, although we do not claim that this captures all intuitively desirable
properties of privacy, we do feel that it allows us to explore the ability to create views that
simultaneously support the strong ability to answer
certain queries in data integration and the strong inability to answer other queries.

\myparagraph{Restrictions and simplifications}
Although the utility and non-disclosure definitions make sense for any queries,
\emph{in this paper we will assume that  $\util$ and $\secret$ are BCQs without constants (abusing notation
by dropping ``without constants'')}.
While we restrict to the case of a single utility query and secret query here,
\emph{all of our results have easy analogs for a finite set of such queries}.
\tocheck

%% file: tools.tex
\subsection{Some tools}
Throughout the paper we rely on two basic tools.

\myparagraph{Canonical views}
Recall that we are looking for views that are useful for answering 
a CQ $\utility$  over a d-schema. The ``obvious''
set of views to try are those obtained by partitioning
the atoms of $\utility$ among sources, with the free variables of the views including
the free variables of $\utility$ and the variables occurring in atoms from different sources.

Given a CQ $Q$ over a d-schema, and a source $s$,
we denote by $\svars(s,Q)$ the variables of $Q$ that appear in an atom from source $s$. 
We also denote by $\sjvars(s,Q)$, the ``source-join variables of $s$'' in $Q$:
the variables that are in $\svars(s,Q)$ and which also occur in
an atom of another source.

\begin{definition}
The \emph{canonical view of $Q$} for source $s$, $\canonview^s(Q)$, has a view definition formed
by conjoining all $s$-source atoms in $Q$ and then
existentially quantifying all bound variables of $Q$ in $\svars(s,Q) \setminus \sjvars(s,Q)$. The \emph{canonical d-view of $Q$} is formed by taking the canonical view for each source.
\end{definition}

In Example \ref{ex:run}, the canonical d-view is what we referred to as ``the obvious
view design''. It would mean that one source
has a  view exposing $\exists \treatdate ~ \treatment(\patientid, \treatmentinfo, \treatdate)$ -- since
$\patientid$ is a source-join variable while $\treatmentinfo$ is a free variable of $Q$.
The other source should expose a view revealing
${\exists \myaddress. \patient(\patientid, \age, \myaddress)}$,
since $\myaddress$ is neither a free variable nor shared across sources.

\myparagraph{The critical instance}
In the definition of UN non-disclosure of a query by a set of views, we required that on \emph{any}
instance of the sources, a user who has access to the views
cannot reconstruct the answer to the query $\secret$.
An instance that will be helpful in several examples is the 
following ``most problematic'' instance \cite{marnettecritinst}.

\begin{definition}
The \emph{critical instance} 
of a schema $\aschema$
is the instance whose
active domain consists of a single element  $\criticalelement$ and whose facts are 
$R(\criticalelement, \ldots, \criticalelement)$ for all relational names 
$R$ in $\aschema$. 
\end{definition}
Note that every BCQ
over the relevant relations holds
on the critical instance of the source.
The  critical instance is the hardest instance for UN non-disclosure
in the following sense:
\begin{theorem} \cite{lics16,privacyjair}
\label{thm:critinst}
Consider any 
CQ views $\V$, and any BCQ
$\secret$. 
If $\secret$ is determined by $\V$ at some instance of the source schema
then it is determined by $\V$ at the critical instance.
\end{theorem}

%% file: noconstraints.tex

\input{minimalquerynoconstraints}

\input{consequencesminimal}

\input{counterexamplenonminimal}

\input{arbviews}

%% file: minimalquerynoconstraints.tex
\section{CQ views} \label{sec:cqviews}
Returning to Example \ref{ex:run}, recall the intuition  that 
the canonical d-view of $\utility$ is the ``least informative d-view''
that supports the ability to answer $\utility$.
We start our analysis by proving such a result,
but with two restrictions: $\utility$ must be a minimal  CQ, and we only consider
 views specified by CQs:

\begin{theorem}\label{thm:minimalqnoconstrdeterminacy}[Minimally informative useful CQ views]
If CQ-based d-view $\Varb$ determines a minimal Boolean CQ $\util$,   
then $\Varb$ determines each canonical view 
 $\canonview^s(\util)$  of $\util$.
\end{theorem}

We only sketch the proof here.
The first step is to show that the determinacy of a CQ $Q$ by a CQ-based d-view $\Varb$ leads to 
a certain homomorphism of $Q$ to itself. This step follows using a characterization
of determinacy of CQ queries by CQ views via the well-known chase procedure (See \cite{ustods} or the ``Green-Red chase'' in
\cite{redspider}).
The second step is to argue that if this homomorphism is a bijection, then it implies that the canonical
d-view determines $Q$, while if the homomorphism is not bijective, we get a contradiction of minimality.
This  step relies on an analysis of how the chase characterization of determinacy factorizes over a  d-schema.


%% file: consequencesminimal.tex
\myparagraph{Consequences} 
Combining Theorem \ref{thm:minimalqnoconstrdeterminacy} and Proposition \ref{prop:minprivate}
gives a partial answer to the question of how to obtain
useful and non-disclosing views:
\begin{corollary} \label{cor:minimalcanonviews} For any non-disclosure function $F$, BCQs $\util$ and $\secret$, if there is
a CQ based d-view that is useful for $\util$ and non-disclosing for $\secret$ according to $F$, then  
the  canonical d-view of $\util^\minq$ is such a d-view, where $\util^\minq$ is any minimal
CQ equivalent to $\util$.
\end{corollary}

If we consider the specific non-disclosure notion, UN non-disclosure, we can infer a complexity bound
from combining these results with prior work on the complexity of checking non-disclosure 
(Theorem 44 of \cite{privacyjair}):

\begin{corollary} \label{cor:sigmatwop} There is a $\sigmatwop$ algorithm
taking as input BCQs $\util$ and $\secret$
 and determines
whether there is a CQ-based d-view  that is useful for $\util$ and UN non-disclosing for $\secret$.
If $\util$ is assumed minimal then the algorithm can be taken to be in $\conp$.
\end{corollary}

\michael{Probably remove for conference version}
We will not focus on the algorithmic consequences more in the body of the paper, but
some further comments can be found in the appendix.

%% file: counterexamplenonminimal.tex
The following example shows that the requirement that $\util$ is minimal (that is, has
no redundant conjuncts) in
Theorem \ref{thm:minimalqnoconstrdeterminacy} is essential.

\begin{example} \label{ex:nonminimal}
Consider two sources. The first source comprises the relations ${R_1}$, $R_2$ and $R_3$, 
while the second source comprises a single relation $T$. 
Consider also the conjunctions of atoms $C_1$ and $C_2$ defined as:
\begin{align}
	C_1 &=  R_1(x,y) \wedge T(x) \nonumber \\
	C_2 &= R_1(x',y') \wedge R_1(y',z')  \wedge R_1(z',x') ~\wedge \nonumber \\ 
	       & ~~~~~~~~~~~T(x') \wedge R_2(y') \wedge R_3(z') \nonumber
\end{align}
The conjunction $C_1$ states that there is an element in $T$ that is the source of an $R$ edge.
The conjunction $C_2$ states that there is 
an $R_1$-triangle with one vertex in $T$, a second in $R_2$, and a third in $R_3$.
Consider now the query $\util$ defined as
\begin{align}
	\exists x,y,x',y',z'. ~ C_1 \wedge C_2 \nonumber
\end{align}
Note that the conjunction of atoms $C_1$ in $\util$ is redundant.
Indeed, the query ${\util^{min} = \exists x',y',z'. C_2}$ is equivalent to $\util$. 

The canonical view of $Q$ for the $R_i$'s source is:
\begin{align}
\exists y, y', z'. R_1(x,y)  \wedge \nonumber \\ 
R_1(x',y') \wedge R_1(y',z') \wedge R_1(z',x') \wedge \nonumber \\ 
R_2(y') \wedge R_3(z') \nonumber
\end{align}

But the canonical view of $\util^{\minq}$ for this source is
\begin{align}
	\exists y', z'. ~ R_1(x',y') \wedge R_1(y',z') \wedge R_1(z', x') \wedge \nonumber \\
	R_2(y')  \wedge R_3(z') \nonumber
\end{align}

Theorem \ref{thm:minimalqnoconstrdeterminacy} tells us that the canonical d-view of $\util^{\minq}$ is a minimally informative  useful d-view
within the class of CQ views.
We claim that the canonical d-view of $\util$ is \emph{not} minimally informative for this class, and in fact
reveals significantly more than the canonical d-view of $\util^{\minq}$.
Consider
 the secret query ${\secret= \exists t.R_1(t,t)}$ stating that there is an $R_1$ self-loop.
The canonical d-view of $\util^{\minq}$ 
is UN non-disclosing for $\asecretquery$. Indeed, given
any instance $\globalinst$, consider the instance $\globalinst'$ in which the $T$ source is identical to the
one in $\globalinst$, but the $R$ source is replaced by one where each  node $e$ in the canonical
view for $\globalinst$ is in a triangle with distinct elements for $y', z'$.
Such a $\globalinst'$ does not satisfy $\asecretquery$, and is indistinguishable from $\globalinst$ according
to the canonical d-view of $\util^{\minq}$.

In contrast, the canonical d-view of $\util$ is UN disclosing for $\secret$. Consider $\globalinst_0$, the critical instance
for the source schema
(see Section \ref{sec:prelims}).
On $\globalinst_0$ the returned bindings have  $x$ as only $\critelement$, and from this we can infer
that the witness elements for $y'$ and $z'$ can only be $\critelement$, and hence the d-view discloses
$\secret$ with $\globalinst_0$ as the witness.

By Proposition \ref{prop:minprivate} the canonical d-view of $Q$ can not be minimally informative
within
the class of CQ views.
\end{example}

%% file: arbviews.tex
\section{Arbitrary views} \label{sec:arbviews}

In the previous section we showed that the  canonical
d-view is minimally informative  within the class of CQ views, assuming
that the utility query is minimized.
We now turn to minimally informative useful views, not restricting
to views given by
CQ view definitions.

Our goal will be to arrive at a  generalization of the notion of canonical d-view that
gives the minimal information over arbitrary d-views that are useful for a given BCQ
$\utility$. That is we want to arrive at
an analog of Theorem \ref{thm:minimalqnoconstrdeterminacy} replacing
``CQ views'' by  ``arbitrary views'' and ``canonical view'' by a generalization.

\subsection{An equivalence class representation of minimally informative views}
Recall that general views are defined by queries, where a query can be
any function on instances.  An \emph{Equivalence Class Representation
  of a d-view} (ECR) consists of an equivalence relation $\equivr_s$
for each source $s$.  An ECR is just another way of looking at a
d-view defined by a set of arbitrary functions on instances: given a
function $F$, one can define an equivalence relation by identifying
two local instances when the values of $F$ are the same. Conversely,
given an equivalence relation then one can define a function mapping
each instance to its equivalence class.
A d-instance $\globalinst$ is indistinguishable from a d-instance $\globalinst'$ by the
 d-view specified  by ECR $\<\equivr_s:s \in \sources\>$ 
	exactly when $\globalinst_{s} \equivr_s \globalinst'_s$ holds for each $s$.
Determinacy of one d-view by another corresponds exactly to the
refinement relation between
the corresponding ECRs. We will thus abuse notation by talking about indistinguishability,  usefulness,
and minimal informativeness of an ECR, referring to the corresponding d-view.

Our first step will be to show that 
there is an easy-to-define ECR whose corresponding d-view is minimally informative.
For a source $s$, an \emph{$s$-context}  is an instance for each source other than $s$.
Given an $s$-context $C$ and an $s$-instance $\inst$, we use $(\inst,C)$ to denote
the d-instance formed by interpreting the $s$-relations as in $\inst$
and the others as in $C$.

We say two $s$-instances $\inst,\inst'$  are \emph{$(s,Q)$-equivalent} if
for any $s$-context $C$, ${(\inst, C) \models Q \Leftrightarrow (\inst', C) \models Q}$.
We say two d-instances  $\globalinst$ and $\globalinst'$ are \emph{globally $Q$-equivalent} if  for each source $s$,
the restrictions of $\globalinst$ and $\globalinst'$ over source $s$, are ($s$, $Q$)-equivalent.

Global $Q$-equivalence is clearly an ECR. Via ``swapping one component at a time'' 
we  can see that the corresponding d-view is useful for $Q$.
It is also not difficult to see that this d-view is minimally informative  for $Q$ within  the class of all views:

\begin{proposition}  \label{prop:qequivcoarse} 
The d-view corresponding to global $Q$-equivalence is a minimally informative useful d-view
for $Q$ within the collection of
all views.
\end{proposition}

Note that the  result can be seen as an analog of 
Theorem \ref{thm:minimalqnoconstrdeterminacy}. From it we conclude
an analog of Corollary \ref{cor:minimalcanonviews}:

\begin{proposition} \label{prop:qequivsuffices}
If there is any  d-view that is useful for BCQ $Q$ and non-disclosing for 
BCQ $\asecretquery$,
then the d-view given by global $Q$-equivalence is useful for $Q$ and non-disclosing for $\asecretquery$.
\end{proposition}

\myparagraph{From an ECR to a concrete d-view}
We now have a useful d-view that is minimally informative within the set of all d-views, but  it is given only as the ECR
global $Q$-equivalence, and it is not clear that there are any views in the usual sense --
isomorphism-invariant functions mapping  into relations of some fixed  schema --- that correspond to this ECR.
Our next goal is to show that global $Q$-equivalence is induced by a d-view defined using  standard
database queries.

A \emph{shuffle} of a CQ is a mapping
from its free variables to themselves (not necessarily injective).
Given a CQ $Q$ and a shuffle $\shuf$,
we denote by $\shuf(Q)$ the CQ that results after replacing each variable occurring in $Q$ by its $\shuf$-image.
We call $\shuf(Q)$ a \emph{shuffled query}. For example, consider the query
${\exists y. R(x_1, x_2, x_2, y) \wedge S(x_2, x_3, x_3, y)}$.
Then, the query ${\exists y. R(x_2, x_1, x_1, y) \wedge S(x_1, x_1, x_1, y)}$ is a shuffle of $Q$. 

The \emph{canonical context query for $Q$ at source $s$},
$\canoncont^{s}(Q)$, is the CQ whose atoms are all the atoms of $Q$
that are \emph{not} in source $s$, and whose free variables are 
$\sjvars(s,Q)$. 

\begin{definition}
For a source $s$, a BCQ $Q$ and a variable binding $\sigma$ for $Q$,
a shuffle $\shuf$ of $\canonview^s(Q)$ is \emph{invariant relative to $\<\sigma, \canoncont^s(Q)\>$} if 
for any d-instance $\inst'$ 
where $\inst', \sigma \models \canoncont^s(Q)$, we have 
$\inst', \sigma \models \shuf(\canoncont^s(Q))$.
\end{definition}

Note that we can verify this invariance by finding a homomorphism from
$\sigma(\shuf(\canoncont^s(Q)))$
to $\sigma(\canoncont^s(Q))$.

Invariance talks about every binding $\sigma$. We would like to abstract
to bindings satisfying a set of equalities.
A \emph{type} for $\sjvars(s,Q)$ is a set of equalities
between  variables in $\sjvars(s,Q)$. The notion of a variable binding satisfying
a type is the standard one.
For a type $\tau$, we can  talk about a mapping $\shuf$ being \emph{invariant relative to $\<\tau, \canoncont^s(Q)\>$}:
the invariance condition holds for all bindings $\sigma$ satisfying $\tau$.

For a source $s$ and a CQ $Q$, let ${\tau_1,\dots,\tau_n}$
be all the equality types over the variables in $\sjvars(s,Q)$ and $\vec x$ be the variables in $\sjvars(s,Q)$.
\begin{definition}
The \emph{invariant shuffle views} of $Q$ for source $s$ is the set of views $V_{\tau_1}, \dots, V_{\tau_n}$ where
each $V_{\tau_i}$ is defined as
${\tau_i(\vec x) \wedge \bigvee \nolimits_{\shuf} \shuf(\canonview^{s}(Q))}$,
where $\vec x$ are the source-join variables of $Q$ for source $s$,
and where the disjunction is  over shuffles invariant relative to $\tau_i$.
\end{definition}

Note that since the domain
of $\shuf$ is finite, there are only  finitely many mappings on them,
and thus 
there are finitely many disjuncts in each view up to equivalence.

We can show that global $Q$ equivalence corresponds to agreement on these views:
\begin{proposition} \label{prop:fullshuffle} For any BCQ $Q$ and any source $s$,
two $s$-instances
are $(s,Q)$-equivalent if and only if they agree on each 
invariant shuffle view
of $Q$ for $s$.
\end{proposition}

Putting together Proposition \ref{prop:qequivsuffices} and \ref{prop:fullshuffle},
we obtain:

\begin{theorem} \label{thm:fullshuffleminimal} [Views that are minimally useful among all views]
The invariant shuffle views are minimally informative for $\util$ within the class of all views.
\end{theorem}

This yields a  corollary  for non-disclosure analogous to Corollary \ref{cor:minimalcanonviews}:

\begin{corollary} \label{cor:fullshufflesuffices} 
If an arbitrary d-view $\varb$ is useful 
for BCQ $Q$ and non-disclosing for BCQ $\asecretquery$, then
the d-view containing, for each
source $s$, the invariant shuffle views of $Q$ for $s$,
is useful and non-disclosing. In particular, some DCQ
is useful for $Q$ and non-disclosing for $\asecretquery$.
\end{corollary}

In Example \ref{ex:run} there are no nontrivial shuffles,
so we can conclude that the canonical d-view is minimally informative within the class of all views.
In general the  invariant shuffle views can be \emph{unsafe}:
different disjuncts may contain distinct variables.
Of course, they can be implemented easily by using
a wildcard to represent elements outside the active domain.
Further, we can convert each of these unsafe views to an ``information-equivalent'' 
set of relational algebra views:

\begin{proposition} \label{prop:makesafe}
For every view defined by a DCQ
(possibly unsafe), there is
a finite set of relational algebra-based views $\V'$ that induces  the same ECR.
Applying this to the invariant shuffle views for a CQ $\util$, we can find a relational algebra-based d-view that is
minimally informative for $\util$ within the class of all views.
\end{proposition}

The intuition behind  the proposition  is to construct  separate views for different
subsets of the variables that occur as a CQ disjunct.
A view with a given set of variables
$S$ will assert that some CQ disjunct with variables $S$ holds and that no disjunct
corresponding to a subset of $S$ holds.

\input{workedexample}

%% file: workedexample.tex
\begin{example} \label{ex:invshuffleviews}
Consider a d-schema with two sources, one containing a ternary relation $R$ and the other containing
a unary relation $S$.
Consider the utility query $Q$:
\begin{align*}
	\exists x_1, x_2, y. R(x_1, x_2, y) \wedge R(y, x_2, x_1) \wedge S(x_1) \wedge S(x_2) \nonumber
\end{align*}
The canonical view for the $R$ source $\canonview^{R}(Q)$ is  ${\exists y. ~ R(x_1, x_2, y) \wedge R(y, x_2, x_1)}$.

Observe that  $\util$ is a minimal CQ, and hence by Theorem \ref{thm:minimalqnoconstrdeterminacy} the canonical d-view of $\util$ is
minimally informative among the CQ-based d-views.
We will argue that this d-view is not minimally informative useful for $\util$ among all d-views, by arguing
that it discloses more secrets than the shuffle views disclose.

Consider the secret query $ \secret = \exists x. R(x,x,x)$.
We show that the canonical d-view of $\util$ is UN disclosing for $\secret$.
Consider the critical instance of the $R$-source. 
An external party will know
that the instance contains $\{R(\critelement,\critelement,y_0),
R(y_0,\critelement,\critelement)\}$ for some 
$y_0$. On the other hand,
if $y_0 \neq \critelement$, then the canonical d-view
would reveal $x_1=y_0, x_2=\critelement$. So $y_0$ must be $\critelement$,
and therefore $\secret$ is disclosed.
By Corollary \ref{cor:minimalcanonviews},
we know that no CQ-based d-view can be UN non-disclosing for $\secret$ and useful
for $\util$. 

The shuffle views of $\util$ are always useful for $\util$.
We will show that they
are UN non-disclosing for $\secret$.
Let us start by deriving the invariant shuffle views for
the $R$ source.
There are two types, $\tau_1$ in which $x_1=x_2$, and $\tau_2$ in which
the variables are not identified.

For a binding satisfying $\tau_1$, the canonical view of $Q$ for the $R$ source is equivalent to 
\[
	\exists y. R(x_1, x_1, y) \wedge R(y, x_1, x_1)
\]
Since
there is only one free variable in it, there is only one invariant
shuffle, the identity.
Thus 
\[
	\view_{\tau_{1}} = \exists y. R(x_1, x_1, y) \wedge R(y, x_1, x_1)
\]
For bindings satisfying $\tau_2$ 
there are several  shuffles invariant for 
${\canoncont^{R}(Q)= S(x_1) \wedge S(x_2)}$:
the identity, the shuffle which swaps $x_1$ and $x_2$,
the shuffle in which $x_1$ and $x_2$ both go to $x_1$, and
the shuffle in which both $x_1$ and $x_2$ go to $x_2$.
Thus we get the view
$\view_{\tau_2}$ defined as $x_1 \neq x_2$ conjoined with:
\begin{align}
\exists y. R(x_1, x_2, y) &\wedge R(y, x_2, x_1) \vee \nonumber \\
\exists y. R(x_1, x_1, y) &\wedge R(y, x_1, x_1) \vee \nonumber \\
\exists y. R(x_2, x_2, y) &\wedge R(y, x_2, x_2) \vee \nonumber \\
\exists y. R(x_2, x_1, y) &\wedge R(y, x_1, x_2) \nonumber 
\end{align}

This last view is unsafe, but via Proposition \ref{prop:makesafe}
we can convert it into a safe relational algebra view
that yields the same ECR, $\view^\safe_{\tau_2}$ defined as $x_1 \neq x_2$ conjoined with:
\begin{align}
	\neg ( \exists y. R(x_1, x_1, y) \wedge R(y, x_1, x_1) ) \wedge \nonumber \\
	\neg ( \exists y. R(x_2, x_2, y) \wedge R(y, x_2, x_2) )  \wedge \nonumber \\
	[( \exists y. R(x_1, x_2, y) \wedge R(y, x_2, x_1) \vee \nonumber \\
	\exists y. R(x_2, x_1, y) \wedge R(y, x_1, x_2) )] \nonumber
\end{align}

We now argue that the shuffle views are UN non-disclosing for $\secret$.
This is because in any d-instance we can replace
each fact $R(x_0, x_0, x_0)$ by facts $R(x_0, x_0, c)$ and $R(c, x_0, x_0)$ for a fresh $c$, obtaining
an indistinguishable instance where $\secret$ does not hold.
Hence by Proposition \ref{prop:minprivate}, the canonical views of $\util$ can not
be minimally informative within the class of all views, or even within the class of relational algebra
views.

\end{example}

%% file: constraints.tex
\section{Local background knowledge} \label{sec:localconstraints}
We now look at the impact of a background knowledge on the sources.
We start with the case of a background theory $\Sigma$ in which
each sentence is \emph{local}, referencing
relations on a single source.

\input{localcqviews}

\input{localarbviews}

\input{replication}

%% file: localcqviews.tex
\subsection{Extension of results on CQ views to local background knowledge} 
It is easy to show that we can not generalize the prior results for CQ views to arbitrary local background
knowledge. Intuitively using such knowledge we can encode design problems for arbitrary views using CQ views.
\input{examplelocalnontgdcqview}

Thus we restrict to \emph{local constraints $\Sigma$ that are existential rules}.
We show that the results on CQ views
extend to this setting.
We must now consider
utility queries $\util$ that are \emph{minimal with respect
to  $\Sigma$}, meaning that there is no strict subquery
equivalent to $\util$ under $\Sigma$. By modifying the chase-based
approach used to prove Theorem \ref{thm:minimalqnoconstrdeterminacy}, we show:

\begin{theorem}\label{thm:minimalqlocalconstrdeterminacy} [Min.Inf. CQ views w.r.t. local rules]
Let $\Sigma$ be a set of local existential rules, 
 $\util$ a CQ minimal with respect to $\Sigma$.
Then the canonical d-view of $\util$ is minimally informative useful within the class of CQ views relative to $\Sigma$.
\end{theorem}

The analog of Corollary \ref{cor:minimalcanonviews} follows from the theorem:
\begin{corollary} \label{cor:minimalcanonviewslocal}
If any CQ based d-view  is useful for $\Sigma$-minimal $\util$ and non-disclosing for $\secret$ relative to $\Sigma$, then
the  canonical d-view of $\util$ is useful for $\util$ and non-disclosing for $\secret$ relative to $\Sigma$.
\end{corollary}

\myparagraph{Consequences for decidability}
Theorem \ref{thm:minimalqlocalconstrdeterminacy} shows that
even in the presence of arbitrary local existential rules $\Sigma$ it suffices
to minimize the utility query under $\Sigma$
and  check the canonical d-view for non-disclosure under $\Sigma$.
For arbitrary existential rules,   even CQ minimization is undecidable. 
But for well-behaved classes of rules (e.g. those
with terminating chase \cite{acyclicity,moreacyclicity}, or frontier-guarded rules \cite{frontier1})
we can perform both minimization and UN non-disclosure checking effectively \cite{ijcai19}.

%% file: examplelocalnontgdcqview.tex
\begin{example} \label{ex:localnontgdcqview}
Consider the schema, utility query $\util$ and secret query $\secret$ from Example \ref{ex:invshuffleviews}.
Let $\Sigma$ consist of the view definitions for the views $VR_1$ and $VR_2$ in the example.
Relative to $\Sigma$ we have CQ view definitions that are useful and UN non-disclosing,
namely the views those simply export $S, VR_1$, and $VR_2$. However, as observed
in Example \ref{ex:invshuffleviews}, the canonical views are UN disclosing for $\secret$.
\end{example}

%% file: localarbviews.tex
\subsection{Extension of results on arbitrary views to the presence of local constraints}\label{subsec:localarb}

We can also extend the results on arbitrary d-views to account for 
local existential rules $\Sigma$. The notion of a shuffle being $\Sigma$-invariant
is defined in the obvious way, restricting to instances that satisfy the constraints in $\Sigma$.
The $\Sigma$-invariant shuffle views are also defined analogously; they are DCQ views, but can be replaced
by the appropriate relational algebra views. Following the prior template, we can show:

\begin{theorem} \label{thm:fullshufflelocalsuffices} [Min.Inf.   views  w.r.t. local rules]
For any set of local existential rules $\Sigma$, the $\Sigma$-invariant shuffle views of $\util$ provide a minimally informative useful d-view for $\util$
within the class of all views, relative to $\Sigma$.
\end{theorem}

The result has effective consequences for ``tame''
rules (e.g. with terminating
chase) with no non-trivial invariant shuffles. In such cases,
the $\Sigma$-invariant shuffle views degenerate to the canonical d-view,
and we can check whether the canonical d-view is non-disclosing effectively
\cite{lics16,ijcai19}.

%% file: replication.tex
\section{Non-local background knowledge} \label{sec:replication}
The simplest kind of non-local constraint is the
replication of a table between sources.
Unlike local constraints, these require some communication among the sources to enforce.
Thus we can consider a replication constraint to be a restricted form of source-to-source
communication.

We will see that several new phenomena  arise in the presence of replication constraints.
Recall that  with only local constraints, we have useful d-views with minimal information.
We can not guarantee the existence of such a
d-view in the presence of replication:
\begin{proposition} \label{prop:nominimalrep}
There is a schema with replication constraint $\Sigma$ and a BCQ $\util$ where
there is no minimally informative useful d-view for $\util$ within the class of all views w.r.t. $\Sigma$.
\end{proposition}
\input{nomin-short}

Given Proposition \ref{prop:nominimalrep}, for the remainder of subsection we will focus on obtaining useful views that
minimize the set of queries that are UN non-disclosed.
If $Q$ is a CQ,
we say that a d-view $\views$ is  \emph{UN non-disclosure minimal for $Q$} within a class
of views $\mathcal C$ if: $\views$ is useful for $Q$
and for any BCQ $\secret$,  if there is a d-view based on $\mathcal C$
which is useful for $Q$ and UN non-disclosing for $\secret$, then $\views$ is UN
non-disclosing for $\secret$
Proposition \ref{prop:minprivate} implies that
if $\views$ is Min.Inf. within $\cal C$, then it is UN non-disclosure minimal
for any BCQ $Q$ 
within $\mathcal C$.

We can use the technique in the proof of Proposition \ref{prop:nominimalrep}
to show a  more promising new phenomenon: there may be d-views that are useful for
CQ $\util$ and
UN non-disclosing for CQ $\secret$, but they are much more intricate
than any query related to the canonical d-view of 
$\util$.
In fact we can show that with a fully-replicated relation in the utility query we can
can get useful and UN non-disclosing  views whenever this is not ruled out for trivial reasons:
\begin{proposition}
  \label{prop:fullrep}
If BCQ $Q$ contains a relation of non-zero arity replicated across all sources
then there is a d-view that is useful for $\util$ and
UN non-disclosing for BCQ $\secret$ if and only if there is no homomorphism of $\secret$ to $\util$. 
Further we can use the same d-view for every  such $\secret$ without a homomorphism into a given $\util$. In particular, there is a view that is UN non-disclosure minimal for $\util$.
\end{proposition}
 Thus even though we do not have
minimally informative useful d-views, we have d-views that are optimal from the perspective of UN non-disclosure and
utility for a fixed $\util$.

We sketch the idea of the proof of Proposition \ref{prop:fullrep}. Given utility query $\util$ and local instance
$\inst$ we can form the ``product instance'' of $\inst$ and $\util$.
The elements of a product
instance will be pairs $(x,c)$ where $x$ is a variable of $\util$ and $c$ an element of $\inst$, and there will be 
atom ${R((x_1, c_1), \dots, (x_n,c_n))}$ in the product exactly when there are corresponding atoms in $\canonview^s(\util)$ and $\inst$.
Thus we will have homomorphisms  from
this instance to both $\inst$ and $\canonview^s(\util)$.
We use ECRs that  make an $s$-instance $\inst$ equivalent to
all instances formed by iterating this product construction of $\inst$ with  $\canonview^s(\util)$.
The d-views corresponding to these ECRs
 will be UN non-disclosing because  in the product
we will have a fresh copy
of any partial match of $\util$, and so the only way for the secret query $\secret$ to hold  in the product will be if
it has a homomorphism into $Q$, which is forbidden by hypothesis.
The replication constraint
 ensures  that if we have two d-instances  that are equivalent, where the interpretation of the replicated relation is non-empty,
the number of iterations of the product construction is the same on each source. Using this fact we can ensure that
the d-views are useful.

\input{example-replication-short}

The views used in Proposition \ref{prop:fullrep}  are not isomorphism-invariant: like the views
from Proposition \ref{prop:nominimalrep}, the product construction
can be seen as applying some value transformation  on the  elements of each instance.
We can show ---  in sharp contrast
to the situation with only local constraints ---  that with replication, even
to achieve this weaker notion of minimality,  it may be essential to use
d-views based on queries that are not isomorphism-invariant.

\begin{proposition} \label{prop:norelalg}
There is a d-schema with a replication constraint, along with
BCQs $\util$ and $\secret$ such that there is a d-view useful
for $\util$ and UN non-disclosing for $\secret$, but there
is no such d-view based on queries returning values in the active domain and commuting with isomorphisms.
\end{proposition}

\input{rep-nomincq}

%% file: nomin-short.tex
Our proof of Proposition \ref{prop:nominimalrep} uses a schema with unary relations $R$, $S$, $T$.
There are two sources: $R$ and $T$ are in different sources,
and  $S$ is replicated between the two sources. Let $\util$ be 
${\exists x. R(x) \wedge S(x) \wedge T(x)}$.

We will explain how our views act when the active domain of our
instances is over the integers. The proof will easily be seen to extend to arbitrary instances (e.g. by having the views
reveal all information outside of the integers).

Consider the functions $F_1(x) = 2x+3$ for $x$ even and
$2x+2$ for $x$ odd.  $F_2(x) = 2x+4$ for $x$ even and $3x$ odd.
Let $\stretchf_1$ be the function that applies $F_1$ to a relation
element-wise, and similarly define $\stretchf_2$ using $F_2$.  Notice
that $F_1$ maps $0$ to $3$ and $1$ to $4$, while $F_2$ maps $0$ to $4$
and $1$ to $3$.

We define a d-view  $\views_1$ via an ECR,
relating two instances $\inst$ and $\inst'$ of the source exactly when $\inst'$ can
be obtained from applying  $\stretchf_1$ on $\inst$ some number of times
(applying it to both relations of the source) or vice versa.
That is, the ECR of $\views_1$ is the smallest equivalence relation containing
each pair $(\inst, \stretchf_1(\inst))$.
Let  $\views_2$ be defined analogously using $\stretchf_2$.
To see that $\views_1$ and $\views_2$
are useful we will use the following claim, which captures their key properties:
\begin{claim}
If we have two  d-instances satisfying the replication constraint,
${(\inst_1, \inst_2)}$ and ${(\inst'_1, \inst'_2)}$, with
the replicated relation instances non-empty, then:
\begin{compactitem}
\item  We cannot have ${\inst'_1=\stretchf^i_1(\inst_1)}$, ${\inst'_2=\stretchf^j_1(\inst_2)}$ 
with ${i \neq j}$; and similarly for $\stretchf_2$.
\item We  cannot have ${\inst'_1=\stretchf^i_1(\inst_1)}$ and ${\inst_2=\stretchf^j_1(\inst'_2)}$
unless $i=j=0$;
and similarly for $\stretchf_2$.
\end{compactitem}
\end{claim}
\begin{proof}
Let $S$ be the content of the replicated relation in
$(\inst_1, \inst_2)$, while $S'$ is the content of the replicated
relation in $(\inst'_1, \inst'_2)$.

We focus first on $\stretchf_1$. For the first item, let $c(S) =
max \{ |x+2| \mid x\in S\}$. We can check directly that for any
non-empty $S$, $c(\stretchf_1(S)) > c(S)$. Then
$\stretchf^i_1(S)=S'=\stretchf^j_1(S)$ implies that $i=j$ as otherwise
$c(\stretchf^i_1(S)) > c(\stretchf^j_1(S))$ when $i>j$ and
$c(\stretchf^i_1(S)) < c(\stretchf^j_1(S))$ when $i<j$. For the second
item we would have, $\stretchf^i_1(S)=S'$ and $\stretchf^j_1(S')=S$
which means $\stretchf^{i+j}_1(S)=S$ which is only possible for
$i+j=0$.

For $\stretchf_2$, the proof is the same, but now using the function $d$
defined as $d(S) = max\{|x+ 1.5| \mid x\in S\}$. We can check that $d(S)
< d(\stretchf_2(S))$ for any non-empty $S$.

\end{proof}

From the claim, usefulness follows easily.  Suppose we have
$(\inst_1, \inst_2)$ satisfying $\util$, and $(\inst'_1, \inst'_2)$ is
equivalent to $(\inst_1, \inst_2)$. From $(\inst_1, \inst_2)$
satisfies $\util$, we know that the replicated relation in $\inst_1$
and $\inst_2$ is non-empty, so the claim applies to tell us that
$\inst'_1=\inst_1$, $\inst'_2=\inst_2$.

Now suppose $\views$ were a minimally informative useful d-view for $\util$.
We must have $\views_1$ and $\views_2$ determine $\views$.
Thus in particular if we have two local instances that agree
on \emph{either} $\views_1$ or $\views_2$, then  they agree on $\views$.

Consider a d-instance $\globinst$ with $R= \{0\}$,  $S=\{0, 1\}$,  $T=\{0\}$, 
and a d-instance $\globinst'$  with $R=\{3\}$, $S=\{3,4\}$, $T=\{4\}$.
These are  both valid instances (i.e. the replication
constraint is respected).
But $\globinst'$ is obtained from $\globinst$ by applying $\stretchf_1$ on the source with $R$, and 
by applying $\stretchf_2$ on the other source. 

Thus $\globinst'$ and $\globinst$ are indistinguishable
by $\views$,
but $\util$ has a match in $\globinst$ but not in $\globinst'$. This contradicts the
assumption that $\views$ is useful for $\util$.

%% file: example-replication-short.tex
\begin{example} \label{ex:nonlocalconstraints}
We give an example of the power of Proposition \ref{prop:fullrep}, and we highlight the difference
from the situation with only local constraints.
Suppose we have two sources, with one binary relation $S$ replicated between the two, and each source
having one non-replicated binary relation, $R$ in one source and $T$ in the other.
The utility query $\util$ 
is ${\exists x,y. R(x,y) \wedge S(x,y) \wedge T(x,y)}$ and the secret query $\secret$ is
$\exists x. R(x,x)$.

Since $\secret$ is not entailed by $\util$,
Proposition \ref{prop:fullrep} implies that there are views that are useful for $\util$ but
UN non-disclosing for $\secret$.
But it is easy to see that the canonical d-view of $\util$ is UN disclosing
for $\secret$.

In fact, one can show that
there are views with these properties that can be defined in relational algebra.
The views are still much more complex than the classes
of candidate views we used earlier.
One can also show that no negation-free views can be useful for $\util$ and UN non-disclosing for $\secret$. See the appendix for more details concerning this example.
\end{example}

%% file: rep-nomincq.tex
Our next observed limitation of Proposition \ref{prop:fullrep}
is that it tells us nothing about Min.Inf. d-views in the presence of replication
within the class of CQ-based d-views.
Unfortunately, there we can get quite strong negative results.

\begin{theorem} \label{thm:nodisclosuremincq}
There is a schema with a replication constraint and a utility query $\util$
such that
 there is no CQ-based d-view $\views$ that  is minimal for UN non-disclosure
within the class of CQ views.
In particular, there is no minimally informative useful d-view within the
class of CQ-based views.
\end{theorem}

We recall that a query $Q$ is \emph{homomorphism invariant} when, for
all instances $\I$ and $\I'$ and all homomorphism $\mu$ from $\I $ to
$I'$ then $\mu(Q(\I)) \subseteq Q(\I')$. The query $Q$
is \emph{$\adom$-based} when $\adom(Q(\I)) \subseteq \adom(I)$.  Note
that all queries defined by CQs, UCQs or Datalog are
homomorphism-invariant and $\adom$ based. A view is
homomorphism-invariant or $\adom$-based exactly when its defining
query is.

Our schema has two sources, $\Pcal$ and $\Scal$ all containing an eponymous
binary relation (respectively $P$ and $S$). Both sources contain a
shared binary relation $T$. 
The utility query is $\util= \exists w,x,y,z . T(x,y) \land S(y,z) \land T(z,w) \land
P(w,x)$.
We consider three secrets:
\begin{itemize}
  \item $\secret_1= \exists x. S(x,x)$,
  \item $\secret_2= \exists x,y. T(x,y) \land S(y,x)$,
  \item $\secret_3= \exists x,y,z. T(x,y)\land S(y,z) \land T(z,x)$.
\end{itemize}

It is easy to see that, for each of these secrets, there exists a
d-view based on CQs that is useful for $\util$ and UN non-disclosing
for the secret. For instance for $\secret_1$ we have the d-view defined
by the queries:
$$Q_{\Scal}(x,w) = \exists y,z. T(x,y)\land S(y,z) \land T(z,w)$$
$$Q_{\Pcal}(w,x) = P(w,x)$$

However, we can then show (see the appendix) that:

\begin{proposition} \label{lem:oneofthese}
Any d-view that is useful for $\util$, homomorphism-invariant, and
$\adom$-based must necessarily be UN disclosing for one of the secrets
among $\secret_1 \dots \secret_3$.
\end{proposition}

Theorem \ref{thm:nodisclosuremincq} follows easily from the results
above, any d-view minimal for UN non-disclosure cannot be
homomorphism-invariant and $\adom$-based.

%% file: conc.tex
\section{Discussion and outlook} \label{sec:conc}
We have studied the ability to design views that satisfy diverse goals:  expressiveness
requirements in terms of full disclosure of
a specified  set of queries in the context of data integration, and inexpressiveness restrictions
in terms of either  minimal utility or minimizing disclosure of 
queries.
Our main results characterize information-theoretically
minimal views that support the querying of a given CQ.

We consider only a limited setting; e.g.  CQs for the utility query.
Our hope is that the work can serve as a basis for further exploration of the trade-offs
in using query-based mechanisms in a variety of settings.

Even in this restricted setting, our contribution focuses primarily on expressiveness, leaving open 
many questions of decidability and complexity. In particular we do not know
whether the $\sigmatwop$ bound of Corollary \ref{cor:sigmatwop} is tight.
Nor do we know whether the analogous question for arbitrary views -- whether
there is an arbitrary  d-view that is useful for a given $\util$ but UN non-disclosing
for a given $\secret$ -- is even decidable. Our results reduce this to a non-disclosure
question for the shuffle views.

Lastly, we mention that our positive results about CQ views (e.g. in Section \ref{sec:cqviews})
rely on an analysis of the chase, which can be infinite. Thus they
are only proven for a semantics of usefulness that considers all instances. We believe that
the analogous results where only finite instances are considered can easily be proven using the techniques of Section \ref{sec:arbviews}, but leave 
this for future work.

%% file: appendix.tex
\noindent {\large\textbf{APPENDIX}}
\section*{Proofs for Section \ref{sec:cqviews}: designing minimally informative CQ-based  views}
\input{app-minimalquery}

\newpage
\section*{Proofs for Section \ref{sec:arbviews}: designing minimally informative arbitrary views}
\input{app-arbviews}

\input{app-safe}

\input{app-decomp}

\newpage
\section*{Proofs for Section \ref{sec:localconstraints}: extensions
in the presence of local constraints} 
\input{app-localcqviews}

\input{app-localarbviews}

\newpage
\section*{Proofs for Section \ref{sec:replication}: balancing
expressiveness and inexpressiveness in the presence of replication
constraints between sources}
\input{app-replicationexample}

\input{app-replication-product}

\input{app-norelalg}

\input{app-repnomincq}

%% file: app-minimalquery.tex
\subsection*{Proof of Theorem \ref{thm:minimalqnoconstrdeterminacy}}

Recall the statement:

\medskip

If CQ  views $\Varb$ determine a minimal Boolean CQ $\util$,  
then $\Varb$ determines each canonical view 
 $\canonview^s(\util)$  of $\util$.

\medskip

We will make use of the following property of minimal queries, which is easy to verify:
\begin{lemma}\label{lem:minimality1}
    If $Q$ is minimal, then there does not exist a homomorphism $h$ from $Q$ into itself that maps two different variables occurring in
    $Q$ to the same variable.
\end{lemma}

We will begin the proof by showing that the determinacy of a CQ $Q$ by a set of
CQ views $\Varb$ leads to the existence
of a certain homomorphism of $Q$ to itself.

Let $Q$ be a Boolean conjunctive query
and $\Varb$ be an arbitrary set of conjunctive query views.
We will need to review
an algorithm for checking determinacy.
We fix a signature for our queries and views, which we refer to as the
\emph{original signature}. From it we derive a \emph{primed signature}, containing
a relation $R'$ for each $R$ in the original signature.
Given a formula $\phi$  in the original signature, we let $\phi'$ be formed
by replacing every relation $R$ in $\phi$ by $R'$.
We use a similar notation for a set of facts $S$ in the original signature.
In particular, for a conjunctive query $Q$ in the original signature, $Q'$ refers to the conjunctive query obtained
by replacing every relation symbol by its primed counterpart.

Given a view  $\view(\vec x)$ defined by conjunctive
query $\exists \vec y ~ \phi(\vec x, \vec y)$, the corresponding
\emph{forward view definition} (for $\view$) is the rule:
\[
\phi(\vec x, \vec y) \rightarrow \view(\vec x)
\]
while the 
\emph{inverse view definition} for $\view$ is the rule:
\[
\view(\vec x) \rightarrow \exists \vec y .  \phi(\vec x, \vec y)
\]
We let $\forwview(\views)$ and $\backview(\views)$,
denote the forward and inverse view definitions for views in $\views$, and
$\forwview'(\views)$ and $\backview'(\views)$ the same axioms but for
primed copies of the base predicates.

\begin{example}
Suppose we have a set of views $\views$ consisting only of view $V(x)$ given by CQ $\exists y,z. R(x,y) \wedge S(x,z)$.
Then $\forwview(\views)$ contains the rule:
\[
 R(x,y) \wedge S(x,z) \rightarrow V(x)
\]
while 
$\backview'(\views)$ contains the rule
\[
V(x) \rightarrow \exists y, z. R'(x,y) \wedge S'(x,z)
\]
\end{example}

Letting $\Sigma_{Q, \V}$ to be the axioms above, we
can easily see that determinacy can be expressed as: 
\[
Q \wedge \Sigma_{Q, \V} \models Q'
\]
This is a containment of CQs under existential rules, and  the algorithm of Figure ~\ref{alg:query-determinacy} checks this via the 
chase procedure, which is complete for such containments (see Section \ref{sec:prelims}).

\begin{small}
\begin{figure}[h]
\caption{Algorithm $\determinacy(Q,\views)$ for checking determinacy}\label{alg:query-determinacy}
\begin{algorithmic}[1]
    \State $F_0 \defeq \canondb(Q)$                                            \label{alg:query-determinacy:init}
    \While{\text{\bf{true}}} ~~~ \% Next 2 lines create view facts from $F_0$ and then primed facts based on views
        \State $F_1 \defeq \chase_{\forwview(\views)}(F_0)$                       \label{alg:query-determinacy:view-forward}
        \State $F_2 \defeq \chase_{\backview'(\views)}(F_1)$             \label{alg:query-determinacy:view-acc-inverse}
        \If{$\exists h: Q' \rightarrow F_2$ mapping each free variable $v$ of $Q'$ into 
$c_v \in \adom(\canondb(Q))$} 
\textbf{return} \textbf{true} \label{alg:query-determinacy:return}
        \EndIf
        \State $F_3 \defeq \chase_{\forwview'(\views)}(F_2)$       ~~~ \%These 2 lines create view facts from primed facts and then original facts based on views             \label{alg:query-determinacy:primed-view-forward}

        \State $F_4 \defeq \chase_{\backview(\views)}(F_3)$  \label{alg:query-determinacy:view-inverse}
        \State $F_5=$ restrict $F_4$ to original schema
        \If{$F_0 \neq F_5$} ~~~~ \% no fixpoint, so repeat with expanded original facts
            \State $F_0 \defeq F_0 \cup F_5$, \label{alg:query-determinacy:reinit}
        \Else \State \textbf{return} \textbf{false}
        \EndIf
    \EndWhile     
\end{algorithmic}
\end{figure}
\end{small}

Intuitively, the algorithm iteratively chases with the axioms above in rounds, checking
for a match of $Q'$ after each round.
 The correctness is implicit in \cite{NSV}; see also \cite{ustods,redspider}.

\begin{theorem} \label{thm:algworks} The algorithm of Figure \ref{alg:query-determinacy} returns true if
and only if $\views$ determines $Q$.
\end{theorem}

Given an arbitrary input query $Q$,
we use a superscript to distinguish among the different sets computed 
during each iteration of the algorithm of Figure ~\ref{alg:query-determinacy}, e.g., 
$F^i_0(Q,\V)$ denotes instance $F_0$ during the $i$-th iteration of 
the algorithm
when run on $Q$ with $\views$.
For any $Q, \views$ such that the algorithm returns $\true$, and any $j \in 0 \ldots 5$, let
 $F^\infty_j(Q,\views)$ denote the content of $F_j(Q,\V)$ at the point where
the algorithm terminates;
this will be  $F^l_j(Q,\views)$ for some $l$.
Note that in line~\ref{alg:query-determinacy:reinit} of the algorithm
we add to each $F^l_0(Q,\V)$ only the subset of ground atoms  from $F^l_4(Q,\views)$
that belong to the original schema, since the other atoms
will be re-derived in the next round.

For any fact $F=R'(\vec c)$ in the primed signature, let $\unprime(F)$ be the corresponding
fact $R(\vec c)$. For a set of facts $S$, we let
 \[
\unprime(S)= \{\unprime(F): F \in S \mbox{ in the primed signature } \}
\]
The following lemma states the easy fact that iterating the ``chase and backchase'' steps 
        in lines   \ref{alg:query-determinacy:view-forward}
        to \ref{alg:query-determinacy:view-acc-inverse} gives a homomorphic pre-image of what
we started with:

\begin{lemma}\label{lem:backward_homomorphism}
For each $l \geq 1$, there is a homomorphism 
$\mu_l:\unprime(F^l_2(Q,\views)) \rightarrow F^l_0(Q,\views)$ that preserves 
all elements of facts in $F^l_0(Q,\V)$.
\end{lemma}
\begin{proof}
    Let $\Varb_{i} \subseteq \Varb$ be the set of all views providing data from 
the $i$-th datasource and let $\view_{i,1} \ldots \view_{i,j_i}$ enumerate the
view predicate names in $\Varb_i$. , where $\view_{i,j}$
has definition $\vec{B}_{i,j}$.
    Consider the  forward view definition, forward transfer as well as, for
 each $\view_{i,j}$,  the inverse view definition  associated 
     going from  $\view_{i,j}$ to the primed relations. That is, the axioms:
    \begin{align*}
        \vec{B}_{i,j}(\vec{x}, \vec y) & \rightarrow \view_{i,j}(\vec{x}) \\ 
        \view_{i,j}(\vec{x}) & \rightarrow \exists \vec{y} ~ \vec{B}'_{i,j}(\vec{x}, \vec{y}) \\ 
    \end{align*}
    where $\vec{B}_{i,j}(\vec{x})$ is some conjunction of atoms from the $i$-th datasource and
    $\{\vec{y}\}$ is disjoint from $\{\vec{x}\}$. 

Each null $\lambda$ appearing in $F^l_2(Q,\V)$ is generated by a chase step associated with
Line  \ref{alg:query-determinacy:view-acc-inverse} of  the algorithm,
triggered by a unique fact $\view_{i,j}(\sigma)$,
where $\sigma$ is a binding of $\vec{x}$,
with $\lambda$ corresponding to some variable $y_\lambda$ in $ \vec{B}'_{i,j}(\vec{x}, \vec{y})$.
For each  fact  $\Varb'_{i,j}(\sigma)$, we know we can extend to a binding
$\sigma'$ for $\vec y$
such that  $\vec{B}_{i,j}(\sigma')$ is in $F^l_0(Q,\V)$. Our homomorphism
will map $\lambda$ to $\sigma'(v_\lambda)$. 
\end{proof}

An analogous reasoning applies to lines \ref{alg:query-determinacy:primed-view-forward}
to  \ref{alg:query-determinacy:view-inverse}:

\begin{lemma}\label{lem:homomorphism_F5_det}
For any CQ $Q$, each $l \geq 1$ there exists a 
homomorphism $\mu_l:F^l_5(Q,\V) \rightarrow \unprime(F^l_2(Q,\V))$ 
that preserves all elements occurring in $\unprime(F^l_2(Q,\V))$.
\end{lemma}

Using the two lemmas above, we can prove by induction on $l$:
\begin{lemma}\label{lem:backward_homomorphism_determinacy}
For any CQ $Q$, CQ views $\views$, and number
$l$,  there exists a homomorphism $\unprime(F^l_2(Q,\views)) \rightarrow \canondb(Q)$ that
 is the
identity
on elements $c_v$.
By priming each relation, we can get $\mu: F^l_2(Q,\views) \rightarrow \canondb(Q')$.
\end{lemma}

The following simple proposition relates the sets produced by the algorithm
on input $Q$ with the sets produced by the same algorithm on
 one of the canonical views $\canonview^s(Q)$.
\begin{lemma} \label{lem:qtocanonv} For any $l \geq 1$ 
and any source $s$, \\
$F^l_0( \canonview^s(Q) ,\Varb)$ is the same as the set of atoms in source
$s$ within $F^l_0(Q,\Varb)$.
\end{lemma}

We are now ready for the proof of Theorem \ref{thm:minimalqnoconstrdeterminacy}:

\begin{proof} 
    Assume that $\Varb$ does not determine  $\canonview^s(\util)$, for some datasource $s$,
    and  that $\Varb$ determines $\util$. From the latter assumption and
Theorem \ref{thm:algworks} we conclude
that there exists a homomorphism $h_1$
    from $\util'$ into $F^\infty_2(\util,\Varb)$.

    By composing $h_1$ with the $\nu$ produced by Lemma~\ref{lem:backward_homomorphism_determinacy} 
    we obtain a  homomorphism $h_\util$ from $\util'$ into itself.
By unpriming, we can change  $h_\util$ to a homomorphism from $\util$ to itself, and
we will sometimes abuse notation by considering $h_\util$ in either way.
Note that $h_1$, and also $h_\util$,  maps each variable $x$ in $\sjvars(s,Q)$ to some element of the form
$c_v$, since  each such $x$ appears in atoms over distinct sources, and only
elements of the form $c_v$ appear in such atoms within $F^\infty_2(\util, \Varb)$.

We consider two cases.

Case 1: $h_\util$ is injective. Thus $h_\util$ is both an injection on $\util'$ and an injection
from source-join variables to the constant corresponding to a source-join variable.
Then $(h_\util)^{-1}$ composed with $h_1$ is
an injective  homomorphism from $\util'$ into $F^\infty_2(\util,\Varb)$ that maps each
$x$ in $\sjvars(s,Q)$ to $c_x$.
By Lemma \ref{lem:qtocanonv} we can identify
the source $s$ atoms of $F^\infty_2(\util, \Varb)$ with the atoms of
$F^\infty_2(\canonview^s(\util), \Varb)$.
Applying this identification and abusing notation as described above we can see $(h_\util)^{-1}$ as a homomorphism
of $\util$ into $F^\infty_2(\canonview^s(\util), \Varb)$. 
that maps each $v$ in $\sjvars(s,Q)$ to $c_v$. But then the algorithm of Figure
\ref{alg:query-determinacy}
would have returned true when applied to the 
views $\Varb$ and  the query $\canonview^s(\util)$. By Theorem \ref{thm:algworks}, this is a contradiction of the fact
that $\Varb$ does not determine  $\canonview^s(\util)$.

Case 2: $h_\util$ is not injective. But now, by
 Lemma \ref{lem:minimality1}, we have a contradiction of minimality.
\end{proof}

%% file: app-arbviews.tex
\subsection*{Proof of Proposition \ref{prop:qequivcoarse}}
Recall the statement:

\medskip

The d-view corresponding to global $Q$-equivalence is a minimally informative
useful d-view for $Q$ within the collection of
all views.

\medskip

\begin{proof}
To see that $Q$-equivalence is useful, consider two d-instances $\globinst$
and $\globinst'$ that are globally $Q$-equivalent. Source-by-source,
we can modify $\globinst$ on a source $s$
to be the same as $\globinst'$ on source $s$, and the results of $Q$ is not changed, by $s,Q$-equivalence.

Assume $\equivr$ is useful and $\globalinst \equivr \globalinst'$.
Fix a context $C$ for source $s$  such that $(\globalinst_s, C) \models Q$.
We will argue that $(\globalinst'_s, C) \models Q$. Note that for
every source $s$, identical instances for source $s$ must be $\equivr_s$, since $\equivr_s$ is an equivalence relation.
 Thus $(\globalinst_s, C) \equivr (\globalinst'_s,C)$.
Since $\equivr$ is useful for $Q$, this means
 $(\globalinst'_s, C) \models Q$ as required.
\end{proof}

\subsection*{Proof of Proposition \ref{prop:fullshuffle}}
Recall the statement:

\medskip

For any BCQ $Q$ and any source $s$,
two $s$-instances
are $(s,Q)$-equivalent if and only if they agree on each
invariant shuffle view
of $Q$ for $s$.

\medskip

The proof will go through an intermediate view, also given by an ECR.

\begin{definition}
Given two $s$-instances $\inst_1$ and $\inst_2$, we say that $\inst_1$ and $\inst_2$ are  \emph{invariant shuffle
equivalent} (relative to $Q$) if:
Whenever  $\inst_1, \sigma$ satisfies $\canonview^{s}(Q)(\vec x)$
then there is some  shuffle $\shuf$ invariant
relative to $\<\sigma, \canoncont^{s}(Q)\>$, such that $\inst_2, \sigma$ satisfies $\shuf(\canonview^{s}(Q))(\vec x)$
and similarly with the role of $\inst_1$ and $\inst_2$ reversed.
\end{definition}

We show:
\begin{proposition} \label{prop:shuffleequivqequiv}
For any source $s$ invariant shuffle equivalence is identical to
 $s,Q$-equivalence.
\end{proposition}

\begin{proof}
First, suppose $\inst_1$, $\inst_2$ are local instances
for source $s$ that are invariant shuffle equivalent, and suppose we have a match
of $Q$ in $(\inst_1,C)$ via $h^{1,C}$. We want to show that there
is a match in $(\inst_2, C)$. 

We know that the variables in $\sjvars(s,Q)$ are mapped by
$h^{1,C}$ into $\inst_1$. Let $h_0$ be the restriction of $h^{1,C}$ to the variables
of $\sjvars(s,Q)$. Then $\inst_1, h_0$  satisfies $\canonview^{s}(Q)$.
Thus by  shuffle equivalence there is some shuffle $\shuf$ invariant
relative  to $\<h_0, \canonview^{s}(Q)\>$,
such that
$\inst_2, h_0 \models  \shuf(\canonview^{s}(Q))$, with witness $h_2$ extending $h_0$.
We also know that $C, h_0$ satisfies  $\canoncont^{S_0}(Q)$, since
$h^{1,C}$ witnesses this as well. 
 Applying the definition of shuffle invariance,
 $C, h_0$ satisfies
 $\shuf(\canoncont^{S_0}(Q))$. Let $h^{\shuf,C}$ be a homomorphism witnessing this.
Note that since $h^{\shuf,C}$ extends $h_0$ and $h_0$ restricts
$h^{1,C}$, $h^{\shuf,C}$ and $h^{1,C}$ agree on their common variables.
Define $h^{2,C}$ by mapping the variables in $\svars(s,Q)$ as in $h_2$,
and those variables  outside of $\sjvars(s,Q)$ as in $h^{\shuf,C}$.
Since these are two compatible homomorphisms,
$h^{2,C}$ witnesses that $(\inst_2,C) \models Q$.
This completes the argument that invariant shuffle equivalence implies global $Q$-equivalence.

We now show that global $Q$-equivalence implies  shuffle equivalence.
Suppose $s$-instances $\inst_1$, $\inst_2$ are globally $Q$-equivalent, and
 $\inst_1, \sigma$ satisfies $\canonview^{s}(Q)(\vec x)$,
We will show that there is a
shuffle $\shuf$, invariant
relative  to $\<\sigma, \canonview^{s}(Q)\>$,
such that $\inst_2, \sigma \models \shuf(\canonview^{s}(Q))(\vec x)$.

Let $C_1$ be the  canonical database of $\sigma(\canoncont^{s}(Q))$.
That is, for each source $s$ other than $s$, we have a fact for each
atom of $s$ atom of $Q$, where each  variable $x$ of $\sjvars(s,Q)$ is replaced
by $\sigma(x)$ and each variable $x$ not in $\sjvars(s,Q)$ is replaced by a fresh 
element $c_x$. 

$Q$ clearly holds in $(\inst_1, C_1)$. So by global $Q$-equivalence, $Q$ holds in $(\inst_2, C_1)$ via
some homomorphism $h$.  Note that in $(\inst_2, C_1)$ the only elements that
are shared between $\inst_2$-facts and $C_1$-facts lie in the range of $\sigma$.
Thus $h$ must map the variables in $\sjvars(s,Q)$ to the image of $\sigma$.
The binding $\sigma$ may not be injective, but we let $\sigma^{-1}$ be ``some inverse''
that is, any  function from the range of $\sigma$ to variables such that for any $c$ in the range of $\sigma$
$\sigma(\sigma^{-1}(c))=c$.
Let $\shuf$ map any variable $x \in \sjvars(s,Q)$  to $\sigma^{-1}(h(x))$.
So $\sigma(\shuf(x))=h(x)$.

We first claim 
that $\shuf$ is  invariant relative  to $\<\sigma, \canoncont^{s}(Q)\>$.
We show this by arguing that $h$ is a homomorphism
from $\sigma(\shuf(\canoncont^{s}(Q)))$ to
$\sigma(\canoncont^{s}(Q))$. By definition
of $\shuf$, we have for each atom $A(x_1 \ldots x_m, y_1 \ldots y_n)$ 
of $\canoncont^{s}(Q)(\vec x)$,
\begin{align*}
A(h(x_1) \ldots h(x_m), h(y_1) \ldots h(y_n)) = \\
A(\sigma(\shuf(x_1)) \ldots \sigma(\shuf(x_m)),  h(y_1) \ldots h(y_n))
\end{align*}
$A(h(x_1) \ldots h(x_m), h(y_1) \ldots h(y_n))$
is in $\sigma(\canoncont^{s}(Q))$ since by assumption $h$ is a homomorphism
from $Q$ to $(\inst_2, C_1)$,  $C_1$ is the canonical database of $\sigma(\canoncont^{s}(Q))$,
and $\inst_2$ is an $s$-instance, and hence cannot
 contain any facts over the relations in $(\shuf(\canoncont^{s}(Q)))(\sigma)$.
Thus 
\[
A(\sigma(\shuf(x_1)) \ldots \sigma(\shuf(x_m)),  h(y_1) \ldots h(y_n))
\] lies in
$\sigma(\canoncont(Q))$ 
 as required.

We next claim that $\inst_2, \sigma \models \shuf(\canonview^{s}(Q))$.
The witness
will be the extension $h'$
of $\sigma$ that maps all variables
in $\svars(s,Q)-\sjvars(s,Q)$ via $h$.

Consider an atomic formula 
$A(x_1 \ldots x_m, y_1 \ldots y_n)$
 of $\canonview^{s}(Q)$, where $\vec x$
are free variables of $\canonview^{s}(Q)$.
Therefore
$A(\shuf(x_1) \ldots \shuf(x_m), y_1 \ldots y_n)$
 is a generic atom of $\shuf(\canonview^{s}(Q))$.
To  argue that $h'$ is a homomorphism that witnesses $\inst_2, \sigma \models \shuf(\canonview^{s}(Q))$,
we need to argue that 
\[
A( \sigma(\shuf(x_1)) \ldots \sigma(\shuf(x_m)), h(y_1) \ldots h(y_n))
\]
holds in $\inst_2$.

But by the definition of $\shuf$, this is equivalent to showing that
\[
A(h(x_1) \ldots h(x_m), h(y_1) \ldots h(y_n) )
\]
holds in $\inst_2$.
But this follows  since
$h$
is a homomorphism of $Q$ into $(\inst_2, C_1)$.
\end{proof}

To complete the proof of Proposition  \ref{prop:fullshuffle}
we show:
\begin{proposition} \label{prop:intermediate}
For any CQ $Q$ and source $s$,
two $s$-instances
are invariant shuffle equivalent if and only if they agree on each 
invariant shuffle view
of $Q$ for $s$.
\end{proposition}

\begin{proof}
We first show that if  $\inst_1$ and $\inst_2$ agree on the  invariant shuffle views
of $Q$, they are invariant shuffle equivalent.
Suppose $\inst_1, \sigma \models \canonview^{s}(Q)$, and let $\tau$ be the type of $\sigma$.
Since the identity is invariant  relative
to  $\tau$, we have
$\inst_1, \sigma$ satisfies $V_\tau$, and thus $\inst_2, \sigma$ must satisfy it.
Therefore there is $\shuf$ that is invariant relative to  $\tau$
such that $\inst_2, \sigma \models \shuf(\canonview^{s}(Q))$.
Since $\tau$ is of type $\sigma$, we have
$\shuf$ is invariant relative to $\<\sigma, \canoncont^s(Q)\>$.
Arguing symmetrically for $\inst_2$, we see that $\inst_1$ and $\inst_2$ are invariant shuffle 
equivalent.

In the other direction, suppose $\inst_1$ and $\inst_2$ are invariant shuffle equivalent.
We will argue that they agree on each  invariant shuffle view $V_\tau$.

Towards that end, suppose
$\inst_1, \sigma \models V_\tau$. That is,
$\inst_1, \sigma \models
\tau \wedge \shuf(\canonview^{s}(Q))$ for some 
$\shuf$ that is invariant relative to $\tau$.
Let $\sigma'$ be the pre-image of $\sigma$ under $\shuf$: that is, the
variable binding defined by
$\sigma'(x)=\sigma(\shuf(x))$. 
Then $\inst_1, \sigma' \models  \canonview^{s}(Q)$ by definition.
Thus by invariant shuffle equivalence, there is $\shuf'$ invariant for
$\sigma', \canoncont^{s}(Q)$ such that $\inst_2, \sigma' \models \shuf'( \canonview^{s}(Q))$.

Let $\shuf''=\shuf' (\shuf)$. We will show that $\shuf''$ witnesses that $\inst_2, \sigma \models V_\tau$.
We first verify invariance:

\begin{claim}
$\shuf''$  is invariant relative to $\tau,  \canoncont^s(Q)$.
\end{claim}
\begin{proof}
Suppose $\sigma_0$ satisfies $\tau$, and $\inst_0, \sigma_0 \models  \canoncont^s(Q)$.
Let $\sigma'_0$ be the  pre-image of $\sigma'$ under $\shuf$.
Note that  a shuffle that is invariant for $\sigma$ must be invariant
for $\sigma_0$, since $\sigma_0$ satisfies all the equalities that $\sigma$ does.
Similarly a shuffle  that is invariant for $\sigma'$ must be invariant
for $\sigma'_0$.
 We can see that  the following chain of implications:

\begin{align*}
\inst_0, \sigma_0 \models \shuf(\canoncont^s(Q))  \mbox{by invariance of $\shuf$ for $\tau$} \\
\inst_0, \sigma'_0 \models \canoncont^s(Q) \mbox{by definition of $\sigma'_0$} \\
\inst_0, \sigma'_0 \models \shuf'(\canoncont^s(Q)) \mbox{by invariance of $\shuf'$ for $\sigma'_0$} \\
\inst_0, \sigma_0 \models \shuf(\canoncont^s(Q)) \mbox{by definition of $\sigma'_0$ again}
\end{align*}

\end{proof}

We now  show that $\inst_2, \sigma$ satisfies the corresponding shuffled query.

\begin{claim}
$\inst_2, \sigma \models  \shuf''( \canonview^{s}(Q))$.
\end{claim}
\begin{proof}

We make the following observation.
For any  instance $\inst$, bindings $\sigma_0$, CQs $R$, and  shuffles $\shuf_0$, 
let $\sigma_1$ be the pre-image of $\sigma_0$ under $\shuf_0$.
Then $\inst, \sigma_0 \models \shuf(R)$ if and only if $\inst, \sigma_1 \models R$.

Let $\sigma''$ be the pre-image of $\sigma$ under $\shuf''$. Note
that $\sigma''$ is also the pre-image of $\sigma'$ under $\shuf'$.

From the observation above, we see  that the following are equivalent:
\begin{align*}
\inst_2, \sigma \models  \shuf( \canonview^{s}(Q)) \\
\inst_2, \sigma'' \models  \canonview^{s}(Q) \\
\inst_2, \sigma' \models \shuf'( \canonview^{s}(Q))
\end{align*}
which gives the proof of the claim.
\end{proof}

 Putting together the two claims, We conclude that
 $\inst_2, \sigma$ satisfies  $V_\tau$ as required, which completes
the proof of  Proposition \ref{prop:intermediate}.
\end{proof}

%% file: app-safe.tex
\subsection*{Proof of Proposition \ref{prop:makesafe}}
Recall the statement:

\medskip

For every view defined by a DCQ
(possibly unsafe), there is
a set of relational algebra views $\V'$ that induces  the same ECR.
Applying this to the invariant shuffle views for a CQ $\util$, we can find a relational algebra-based minimally informative useful d-view 
for $\util$ within the class of all views.

\medskip

Putting the conclusion ``same ECR'' another way: if $V$ is the original
DCQ view, then we obtain a finite set of views $\V'$ with  the definition of each
view in relation algebra, such that
$\V'$ determines $V$ and $V$ determines $\V'$.

Clearly if we have this for a single DCQ view $V$, we obtain it for a finite set of views
(and hence for a d-view) by applying the construction to each view in the set.

We consider a  DCQ $V(x_1 \ldots x_n)$ 
defined by $\bigvee_i \phi_i$.
For each $\phi_i$ let $\dvars_i$ be the set of variables within them,
and for each subset $S$ of the vars let $D_S$ be the set of $i$ such that
$\phi_i$ uses variables $S$.

Given a set of variables $S=x_{j_1} \ldots x_{j_k}$ with $D_S \neq \emptyset$,
create a view $V_S(x_{j_1} \ldots x_{j_k})$
defined by
\begin{align*}
\bigvee_{i \in D_S} \phi_i(x_{j_1} \ldots x_{j_k}) \wedge
\neg (\bigvee_ {S' \subsetneq D_S, \dvars(\phi_j)=S'} \phi_j)
\end{align*}

\begin{example}
We explain the construction of relational algebra views by example.
Suppose we have a view $V$ given by a DCQ:
\[
R(x,y,z) \vee P(x,y,z) \vee W(x,y, w) \vee T(x,y)
\]

We  have three sets $S$ such that $D_S \neq \emptyset$:
$S_1=\{x,y,z\}, S_2=\{x,y,w\}$ and $S_3=\{x,y\}$.

Our construction will  create views for each of these.

$V_{S_1}(x,y,z)$ is defined by query:
\[
[R(x,y,z) \vee P(x,y,z)] \wedge \neg T(x,y)
\]
$V_{S_2}(x,y)$ is defined by query:
\[
W(x,y) \wedge \neg T(x,y)
\]
Finally, $V_{S_3}(x,y)$ is defined by the query $T(x,y)$.

It is not difficult to see that these views determine $V$ and vice versa.
\end{example}

Returning to the general case,
we claim that the set of views $V_S$ determines $V$ and vice versa.

In one direction suppose $\inst_1$ and $\inst_2$ agree on each $V_S$, and
$\inst_1 \models V(\vec t)$. Choose $i$  such that $\inst_1 \models \phi_i(\vec t)$ with
$S_i=\dvars_i$ minimal. Let $\vec t'$ be the subtuple
of $\vec t$ corresponding to the variables of $\phi_i$.
Then $\inst_1 \models V_{S_i}(\vec t')$ and hence $\inst_2 \models V_{S_i}(\vec t')$.
From this we see that $\inst_2 \models V_S(\vec t)$.

In the other direction, suppose $\inst_1$ and $\inst_2$ agree on $V$, and
$\inst_1 \models V_S(\vec t)$. Fix $\phi_i$ with variables from $S$
such that $\inst_1 \models \phi_i(\vec t)$.  We need to show $\inst_2 \models V_S(\vec t)$. We can
assume by induction that $V$ determines $V_{S'}$
for each $S'$ that is a proper subset
of $S$. 
First consider the case where
 $S$ consists of all variables.
Then  $\inst_1 \models V(\vec t)$ hence $\inst_2 \models V(\vec t)$, and
thus there is some $j$ such that $\inst_2 \models \phi_j(\vec t)$.
If $\phi_j$ contains all the variables of $\phi_i$. Using
the induction hypothesis and $\phi_j$
we can conclude that $\inst_2 \models V_S(\vec t)$ as required.
If $\phi_j$ contains  a proper subset $S'$ of the variables in $S$,
then we have $\inst_2 \models \phi_k(\vec t')$ for $\vec t'$
a proper subtuple of $\vec t$.  Choose $\phi_k$ and $\vec t'$
with this property such that the variables $S'$ involved are minimized.
The $\inst_2 \models V_{S'}(\vec t')$ so by the induction hypothesis
$\inst_1 \models V_{S'}(\vec t')$, which contradicts the facts that
$\inst_1 \models V_S(\vec t)$.

Next consider the case where $S$ is a proper subset of the variables.
We extend $\vec t$ to   $\vec t'$ choosing elements 
outside the
active domain of both $\inst_1$ and $\inst_2$. $\inst_1 \models V(\vec t')$,
and $\inst_2 \models V(\vec t')$. Thus we have a proper subtuple $\vec t''$ of
$\vec t'$ and a disjunct $\phi_k$ such $\inst_2 \models \phi_i(\vec t'')$.
As above, we can choose $\vec t''$ minimal.
By our choice of the elements in $\vec t' -\vec t$, we must have
$\vec t''$ a subtuple of $\vec t$. If $\vec t''=\vec t$, then we can
conclude that $\inst_2 \models V_S(\vec t)$ as required. If $\vec t''$
is a proper subtuple, we argue by contradiction of the induction hypothesis
as above.

\myeat{
}

%% file: app-decomp.tex
\subsection*{Reducing the complexity of the CQ view design  problem}
\michael{Is this interesting enough to include?}
In the body of the paper we showed that to test for useful and UN non-disclosing
CQ views, when our utility query $\util$ is a minimal CQ, we need only
test that the canonical d-view is safe. This requires performing a specialized ``disjunctive chase
with constant-equality EGDs'' (described in \cite{lics16,ijcai19}), after which we check
that $\util$ holds on each instance produced by the chase process. This give
the bound in Corollary  \ref{cor:sigmatwop}.

We now mention some additional insights that can help optimize
this algorithm.
Throughout this subsection, we rely on the critical instance method, and in particular
on Theorem \ref{thm:critinst}. 

The canonical views are a way of decomposing a utility query.
But it turns our that we can also apply the canonical views as a way of decomposing a
secret query.
We start with the following observation, which states that in analyzing
secrecy of a set of CQ views, we can break up the secret query into its canonical views
and analyze them one at a time:


\begin{proposition} \label{prop:decomp} For any Boolean CQ $\secret$ and CQ views $\Varb$,  $\Varb$ is  UN non-disclosing for
$\secret$ if and only if $\Varb$ is UN non-disclosing for $\canonview^s(\secret)$ for
some $s$.
\end{proposition}

\begin{proof}
In one direction, suppose there is $s$
such that $\Varb$ is UN non-disclosing for $\canonview^s(\secret)$. 
We show
that $\Varb$ is UN non-disclosing for $\secret$ by showing
that the critical instance $\globalinst_0$ for the d-schema is
$\Varb$-indistinguishable from some other d-instance where
$\secret$ fails.
We know that  there is some $s$-instance $\inst_s$ 
that is $\Varb^s$ indistinguishable from $\globalinst_s$ but in which
$\canonview^s(\secret)(\critelement)$ does not hold.
Let $\globalinst'$  be formed from taking $\inst_s$ on  source $s$ and
taking the critical instance on the other sources.
choosing the other components arbitrarily. Clearly
$\globalinst'$ is $\Varb$-indistinguishable from $\globalinst$. If $\secret$
held on $\globalinst'$, the only possible witness would be the critical tuple,
since this is the only binding. But the critical tuple fails the conjuncts on source $s$.

In the other direction, suppose $\Varb$ is UN disclosing for $\canonview^s(\secret)$ for some $s$.
We will show that $\Varb$ is UN disclosing for $\secret$.
We know that for each $s$, letting $\inst_s$ be  the critical instance for source $s$, if we apply
the disjunctive chase procedure from \cite{ijcai19,lics16} to  $\inst_s$ then $\canonview^s(\critbinding)$ holds.
But when we apply the disjunctive chase procedure  to the critical
instance for the d-schema, this is the same as applying it
to each component. Thus $\secret$ holds with $\critbinding$
as a witness, so $\Varb$ is UN disclosing for
$\secret$.
\end{proof}

The result above is about a fixed set of CQ views.
But using Corollary \ref{cor:minimalcanonviews} we can  lift
it to an observation about decomposing the secret query
in searching for  the existence of useful and UN non-disclosing views:

\begin{proposition}  \label{prop:decompsecret}
For any Boolean CQs $\util$  and $\secret$, there
are CQ views that are useful for $\util$ and UN non-disclosing
for $\secret$ if and only if for some $s$, the canonical d-view of $\util$
are UN non-disclosing for $\canonview^s(\secret)$.
\end{proposition}
\begin{proof}
If there are CQ views that are useful for $\util$ and
UN non-disclosing for $\secret$, then the canonical d-view of $\util$ are such
a set of views, by Corollary \ref{cor:minimalcanonviews}. Thus by the previous proposition,
they are UN non-disclosing for some  $\canonview^s(\secret)$.

In the other direction, if the canonical d-view of $\util$
is  UN non-disclosing for $\canonview^s(\secret)$ for some $s$, then by the proposition
above they are UN non-disclosing for $\secret$, and thus these views serve as a witness.
\end{proof}

Proposition \ref{prop:decompsecret} implies that to test for useful and UN non-disclosing
views, we need only take each source $s$ and test whether
 $\canonview^s(\util)$ for source $s$ is UN non-disclosing for the  canonical view of $\canonview^s(\secret)$
for source $s$. A witness to failure of
such a test  requires first a deterministic computation that consists of chasing forward
and backward with the view definitions, then a series of guesses of homomorphism
of $\canonview^s(\util)$, followed by the guess of a homomorphism of  $\canonview^s(\secret)$,
thus $\conp$ in the maximum cardinality over all
$s$ of $\canonview^s(\secret)$ and $\canonview^s(\util)$.

%% file: app-localcqviews.tex
\subsection*{Proof of Theorem \ref{thm:minimalqlocalconstrdeterminacy}}
Recall the statement:

\medskip

Let $\Sigma$ be any set of TGDs that are local.
Suppose that $\util$ is a minimal CQ with respect to $\Sigma$.
If CQ views $\Varb$ determine CQ $\util$ over all 
instances satisfying $\Sigma$, then $\Varb$ determines each canonical view 
$\canonview^s(Q)$ of $\util$ over all instances satisfying $\Sigma$.

\medskip

We proceed as in the case of no constraints. 
We modify  the determinacy algorithm by chasing with the local constraints in each round, 
giving us the algorithm in Figure \ref{alg:query-determinacy-constraints}.

\begin{figure}[h]
\caption{Algorithm $\determinacy(\util,\views, \Sigma)$ for checking determinacy with respect to existential rules}\label{alg:query-determinacy-constraints}
\begin{algorithmic}[1]
    \State $F_0 \defeq \chase_{\Sigma}(\canondb(\util))$                                            \label{alg:query-determinacy-constraints:init}
    \While{\text{\bf{true}}}                                         
        \State $F_1 \defeq \chase_{\forwview(\views)}(F_0)$                       \label{alg:query-determinacy-constraints:view-forward}
        \State $F_2 \defeq \chase_{\backview'(\views)}(F_1)$             \label{alg:query-determinacy-constraints:view-acc-inverse}
         \State $G_{2} \defeq \chase_{\Sigma'}(F_2)$
        \If{$\exists h: \util' \rightarrow G_2$ mapping each free variable $v$ of $\util'$ into 
$c_v \in \adom(\canondb(\util))$}
            \State \textbf{return} \textbf{true}                                     \label{alg:query-determinacy-constraints:return}
        \EndIf
        \State $F_3 \defeq \chase_{\forwview'(\views)}(G_2)$                   \label{alg:query-determinacy-constraints:primed-view-forward}
        \State $F_4 \defeq \chase_{\backview(\views)}(F_3)$                       \label{alg:query-determinacy-constraints:view-inverse}
        \State $G_{4} \defeq \chase_{\Sigma}(F_4)$                       \label{alg:query-determinacy-constraints:applysigma}
       \State $G_{5} \defeq$ \mbox{restrict } $G_4$ \mbox{ to the original signature}
        \If{$F_0 \neq G_5$}
            \State $F_0 \defeq F_0 \cup G_5$                            \label{alg:query-determinacy-constraints:reinit}
        \Else
            \State \textbf{return} \textbf{false}
        \EndIf
    \EndWhile     
\end{algorithmic}
\end{figure}

As with  the algorithm in Figure \ref{alg:query-determinacy},
it is a straightforward exercise to see that this algorithm correctly
checks determinacy:

\begin{theorem} \cite{ustods} \label{thm:localalgworks} The algorithm in Figure \ref{alg:query-determinacy-constraints} returns true if and only if 
$\views$ determines $Q$ relative to $\Sigma$.
\end{theorem}

Note that the chase may not terminate, thus this is only a semi-decision procedure.

As before, we have a homomorphism from the output of the algorithm
to its input:

\begin{lemma}\label{lemma:localconstraints_backward_homomorphism}
For any CQ $Q$, views $\V$ and number
$l$,  there exists a homomorphism $\nu:\unprime(F^l_2(Q,\views)) \rightarrow \canondb(Q)$ that is the
identity
on elements $c_v$.
In particularly, if the algorithm in Figure \ref{alg:query-determinacy-constraints} returns $\true$,
 there is a homomorphism from $F^\infty_2(Q,\views)$ to $\canondb(Q)$ that is the
identity
on elements $c_v$.
\end{lemma}

We have the same relation of the output of the algorithm
when run on the canonical view
to the output when run on $\util$:
\begin{lemma} \label{lem:qtocanonvconstraints} For any $l \geq 1$ 
the source $s$ atoms in $F^l_0(\util,\Varb)$
are the same as the atoms of $F^l_0(\canonview^s(Q),\Varb)$.
\end{lemma}

We can  now complete
 the proof of Theorem \ref{thm:minimalqlocalconstrdeterminacy}
as in the case without constraints. Locality of constraints
ensures that the homomorphism
$h_1$ maps each variable $x$ in $\sjvars(s,Q)$ to some element of the form
$c_v$.

%% file: app-localarbviews.tex
\subsection*{Proof of Theorem  \ref{thm:fullshufflelocalsuffices}}
Recall the statement:

\medskip

For any set of local existential rules $\Sigma$, the $\Sigma$-invariant shuffle views of $Q$ 
provide
a minimally informative useful d-view
within the class of all views, relative to $\Sigma$.

\medskip

We can generalize the ECR global $Q$-equivalence
to global $Q$-$\Sigma$-equivalence, looking only at contexts that satisfy $\Sigma$.
Using the same argument we see that the views corresponding to this ECR
is  a minimally informative useful d-view for $Q$ within the class of all views,
relative to $\Sigma$.

Let $\sigma$ be a mapping of $\sjvars(s)$ into some instance
$\inst$.
A shuffle $\shuf$ of $\canonview^s(Q)$ is
\emph{$\Sigma$-invariant relative to $\<\sigma, \canoncont^s(Q)\>$}
if whenever $\inst',\sigma \models \canoncont^s(Q)$  and $\inst'$ satisfies $\Sigma$
then $\inst', \sigma \models \shuf(\canoncont^s(Q))$.
Note that since the rules are local, the notion of a context satisfying them is well-defined.
Invariance is decidable whenever query containment for CQs under $\Sigma$ is decidable;
for example, this is the case when $\Sigma$ is a set of dependencies
with  terminating chase.

Fixing $\Sigma$ and  two $s$-instances $\inst_1$ and $\inst_2$, we say that $\inst_1$ and $\inst_2$ are  \emph{$\Sigma$-invariant shuffle
equivalent} if
whenever  $\inst_1,  \sigma$ satisfies $\canonview^{s}(Q)(\vec x)$ 
then there is some  shuffle $\shuf$ which is $\Sigma$-invariant
relative to $\<\sigma, \canonview^{s}(Q)\>$, such that:
$\inst_2, \vec \sigma \models \shuf(\canonview^{s}(Q))(\vec x)$ \\
and vice versa.

We can now extend Proposition \ref{prop:shuffleequivqequiv}, following the same proof:

\begin{proposition} \label{prop:shuffleequivqequivlocal}
Suppose $\Sigma$ consists of local existential rules.
 Then for all instances satisfying the rules
$\Sigma$-invariant shuffle equivalence is identical to
global $Q$-$\Sigma$ equivalence.
\end{proposition}

\begin{proof}
First, suppose $\inst_1$, $\inst_2$ satisfy all local rules and
are $\Sigma$-invariant shuffle equivalent. Consider a context $C$ that satisfies local rules
$\Sigma$
and suppose we have a match
of $Q$ in $(\inst_1,C)$ via $h^{1,C}$. We want to show that there
is a match in $(\inst_2, C)$.  This will be exactly as in the case without constraints.

We know that the variables in $\sjvars(s,Q)$ are mapped by
$h^{1,C}$ into $\inst_1$. Let $h_0$ be the restriction of $h^{1,C}$ to the variables
of $\sjvars(s,Q)$. Then $\inst_1, h_0$  satisfies $\canonview^{s}(Q)$.
Thus by  $\Sigma$-invariant shuffle equivalence there is a shuffle $\shuf$ which is $\Sigma$-invariant
relative  to $\<h_0, \canonview^{s}(Q)\>$,
such that
$\inst_2, h_0 \models  \shuf(\canonview^{s}(Q))$, with witness $h_2$ extending $h_0$.
We also know that $C, h_0$ satisfies  $\canoncont^{s}(Q)$, since
$h^{1,C}$ witnesses this as well.
 Applying the definition of $\Sigma$-invariance, and noting that
$C$ satisfies $\Sigma$ by assumption, we infer that
 $C, h_0$ satisfies
 $\shuf(\canoncont^{s}(Q))$. Let $h^{\shuf,C}$ be a homomorphism witnessing this.
Note that since $h^{\shuf,C}$ extends $h_0$ and $h_0$ restricts
$h^{1,C}$, $h^{\shuf,C}$ and $h_0$ agree on their common variables.
Define $h^{2,C}$ by mapping the variables in $\svars(s,Q)$ as in $h_2$,
and those variables  outside of $\sjvars(s,Q)$ as in $h^{\shuf,C}$.
Since these are two compatible homomorphisms,
$h^{2,C}$ witnesses that $(\inst_2,C) \models Q$.
This completes the argument that $\Sigma$-invariant shuffle equivalence implies $Q$-$\Sigma$-equivalence.

We now show that global $Q$-$\Sigma$-equivalence implies  $\Sigma$-invariant shuffle equivalence.
Suppose $s$-instances $\inst_1$, $\inst_2$ are $Q$-$\Sigma$-equivalent, and
 $\inst_1, \sigma$ satisfies $\canonview^{s}(Q)(\vec x)$,
We will show that there is a
shuffle $\shuf$, $\Sigma$-invariant
relative  to $\<\sigma, \canoncont^{s}(Q)\>$
such that $\inst_2, \sigma \models \shuf(\canonview^{s}(Q))(\vec x)$.

Let $C_1$ be the context defined in two steps.
We first proceed as in the case
without  constraints:
for each source $s$ other than $s$, we have a fact for each
$s$ atom of $Q$, where each  variable $x$ of $\sjvars(s,Q)$ is replaced
by $\sigma(x)$ and each variable $x$ not in $\sjvars$ is replaced by a fresh element
$c_x$.  In the second step, we perform
the chase construction with $\Sigma$ to get an instance that satisfies the local constraints.

$Q$ clearly holds in $(\inst_1, C_1)$. So by $Q$-$\Sigma$-equivalence, $Q$ holds in $(\inst_2, C_1)$ via
some homomorphism $h$.
As before, the only elements  shared between $\inst_2$ and $C_1$ lie in the range of $\sigma$.
Thus $h$ must map the variables in $\sjvars(s,Q)$ to the image of $\sigma$.
For each $c$ in the image of $\sigma$, choose a variable $v_c$ such that
$\sigma$
maps $v_c$ to $c$. 
Let $\shuf$ map any variable $x \in \sjvars(s,Q)$  to $v_{h(x)}$.
Thus $\sigma(\shuf(x))=h(x)$.

We claim
that $\shuf$ is  $\Sigma$-invariant relative  to $\<\sigma, \canoncont^{s}(Q)\>$.
We show this by arguing that $h$ is a homomorphism
from $\shuf(\canoncont^{s}(Q))(\sigma)$ to
the chase under $\Sigma$ of $\canoncont^{s}(Q)(\sigma)$. By definition
of $\shuf$, we have for each atom $A(x_1 \ldots x_m, y_1 \ldots y_n)$
of $\canoncont^{s}(Q)(\vec x)$,
\begin{align*}
A(h(x_1) \ldots h(x_m), h(y_1) \ldots h(y_n)) ~ = ~
A(\sigma(\shuf(x_1)) \ldots \sigma(\shuf(x_m)),  h(y_1) \ldots h(y_n))
\end{align*}
$A(h(x_1) \ldots h(x_m), h(y_1) \ldots h(y_n))$
is in $\chase_\Sigma(\canoncont^{s}(Q)(\sigma))$ since $h$ is a homomorphism into
$(\inst_2, C_1)$ and thus for facts in $\shuf(\canoncont^{s}(Q))(\sigma)$, it must
map into $C_1$.

Thus we conclude
\begin{align*}
A(\sigma(\shuf(x_1)) \ldots \sigma(\shuf(x_m)),  h(y_1) \ldots h(y_n))
\in \chase_\Sigma(\canoncont^{s}(Q)(\sigma))
\end{align*}
This completes the proof that $h$ is a homomorphism into the chase, and
thus the proof that  $\shuf$ is  $\Sigma$-invariant relative  to $\<\sigma, \canoncont^{s}(Q)\>$.

We next claim that $\inst_2, \sigma \models \shuf(\canonview^{s}(Q))$.
The witness
will again be the extension of $\sigma$ that maps all variables
in $\svars(s,Q)-\sjvars(s,Q)$ via $h$.

Consider an atomic formula $A(x_1 \ldots x_m, y_1 \ldots y_n)$ of $\canonview^{s}(Q)$ where $\vec x$
are free variables of $\canonview^{s}(Q)$.
That is,
\[
A(\shuf(x_1) \ldots \shuf(x_m), \vec y)
\]
is a generic atom of $\shuf(\canonview^{s}(Q))$.
We know that $A(h(x_1) \ldots h(x_m), h(y_1) \ldots h(y_n)))$ holds in $\inst_2$, since $h$
is a homomorphism into $\inst_2$.
Thus
\[
A(v_{h(x_1)},  \ldots, v_{h(x_m)}, h(y_1) \ldots h(y_n)))
\]
 holds of
$\sigma$ in $\inst_2$,
by definition of $v_c$. From this we see that
\begin{align*}
A(\shuf(x_1) \ldots \shuf(x_m), h(y_1) \ldots h(y_n)))
\end{align*}
 holds of $\sigma$ in $\inst_2$
as required.
\end{proof}

Recall that the $\Sigma$-invariant shuffle views of $Q$ for $s$ and types $\tau$ are defined
analogously to the case without local rules, as
$\tau(\vec x) \wedge \bigvee_{\shuf} \shuf(\canonview^{s}(Q))$
where the disjunction is over  $\Sigma$-invariant shuffles of $\tau$.

The following result is proven exactly as in the case without background knowledge:

\begin{proposition} \label{prop:fullshufflelocal} For any Boolean CQ $Q$, and any source $s$
two $s$-instances
are $\Sigma$-invariant shuffle equivalent if and only if they agree on each
$\Sigma$-invariant shuffle view of $Q$ for $s$.
\end{proposition}

Putting the prior results together gives us the extension of
Theorem \ref{thm:fullshufflelocalsuffices}.

%% file: app-replicationexample.tex
\subsection*{The power of replication: more detail on Example \ref{ex:nonlocalconstraints}}
In the body of the paper, we considered Example \ref{ex:nonlocalconstraints},
where there are two sources and
three binary relations $R, S, T$.
$R$ and $T$ are on different sources, while
the background theory $\Sigma$
asserts that  relation $S$ is
replicated. We considered CQs:
\begin{align*}
\util=\exists x ~ y ~ R(x,y) \wedge S(x,y) \wedge T(x,y) \\ 
\secret= \exists x ~ R(x,x) \\
\end{align*}

We mentioned in the body of the paper that by Proposition \ref{prop:fullrep}, 
there is a d-view that is useful for $\util$ and UN non-disclosing for $\secret$, by
making use of the replication constraint.
But the views in the d-view produced by the proposition  are not isomorphism-invariant.
We know from Proposition \ref{prop:norelalg} that it may be necessary
to use views that are not isomorphism-invariant. 

We show that for this particular example there do indeed
exist relational algebra  views  that were useful for $Q$
and UN non-disclosing
for $\secret$ in this example. Thus in exploiting replication we can
sometimes stay within a standard class of views.
We now explain how to achieve this.


Given an instance of the $R$ source $\inst_R$, we say an \emph{$R$-harmless pair} is any
pair of nodes $(x_1, x_2)$, where:
\begin{itemize}
\item $x_1$ and $x_2$ are connected by both $R$ and $S$ edges
\item there is no $S$ self-loop on $x_1$
\item $x_2$ has no outgoing $S$ edges, and $x_2$ is the unique element
that is a target of  an  $S$ edge from $x_1$ with no outgoing $S$ edges.
\end{itemize}

The \emph{modification} of such a pair $(x_1, x_2)$ is the pair $(x_1, x_1)$.

Our view on the $R$-source takes as input an instance $\inst$ of the $R$ source, and returns
all pairs that are modifications of $R$-harmless pairs, unioned with pairs that are in $S \cap R$ but are not $R$-harmless.

A $T$-harmless pair in the $T$-source is defined similarly but replacing $R$ with $T$.
A modification of such a pair is as above.
Analogously, our view on the $T$-source  returns
all pairs that are modifications of $T$-harmless pairs, unioned with pairs that are in $S \cap T$ but are not harmless.
It is clear that  these views can be expressed in relational algebra.

We first show that the views are UN non-disclosing for $\secret$. Consider an instance $\globalinst=(R, T, S)$ where $\secret$ holds, with $S$ the shared relation.
We will  construct an instance
$\globalinst'=(R', T', S')$ with the same view images, but where $\secret$ does not hold.
We let $V^R_\inst$ be the content for the view for the $R$-source on instance $\inst$, and similarly $V^S_\inst$.

We first describe the shared relation $S'$. 
For each element $v$ in either view image ($V_R$ or $V_S)$, we create an $S$ edge to a new element $n_v$
 We also include all pairs
in either view that are not self-loops.

We now describe $R'$. It includes all edges in the view $V_R$ that are not self-loops. It also
contains  an edge from $v$ to $n_v$ if $(v,v)$  is  in $V_R$.
$T'$ is defined analogously.

It is clear that the new instance does not satisfy $\secret$. We
need to show that it agrees with $\globalinst$ on each view.
Note that all of the pairs  $(v,n_v)$ such that $(v,v)$ is in $V_R$ are $R$-harmless. 
A pair of the form $(c, v)$ where $v$ is one of the original nodes of the instance,
cannot be $R$-harmless, since $v$ has an outgoing $S$ edge. Pairs of the form $(c, n_v)$ where $c \neq v$ are not $R$-harmless because there
is no $S$ or $R$ edge between them. Thus the $R$-harmless pairs are exactly those of the form  $(v,n_v)$ where $(v,v) \in V_R$. But the view
will produce all such pairs $(v,v)$. We can conclude
that the views $V_R$ agree on pairs of the form $(v,v)$.

We now consider pairs in the view $V^R_\globalinst$ of the form $(c,d)$ with $d \neq c$.  By definition, such pairs are included in both $S'$ and $R'$.
They are not $R$-harmless in $\globalinst'$, since $d$ has an outgoing edge to $n_d$. 
Thus they are included in $V^R_{\globalinst'}$. Conversely, if a pair $(c,d)$ with $c \neq d$ does
not  occur in $V^R_\globalinst$, we can argue that it is not in $V^R_{\globalinst'}$. This follows since we will not have $R'$ hold of $(c,d)$.

We next show that the views are useful for $Q$, by arguing that
$Q$ can be rewritten as the intersection of $V_R$ and $V_S$.

In one direction, we suppose $Q$ has a match $(c,d)$ in an instance $\globalinst$, and
we argue that $(c,d)$ is in the intersection of $V_R$ and $V_S$ 
evaluated on $\globalinst$.
Note that such a pair is $R$-harmless if and only if it is $T$-harmless because
$T \wedge R$ holds of it in $\globalinst$.
Thus we will distinguish pairs that are harmless (meaning $R$ or $T$-harmless) versus
pairs that are not harmless. If $(c,d)$ is a harmless pair, then $(c,c)$ will be in the intersection of $V_R$ and $V_S$ within $\globalinst$.
While if $(c,d)$ is not a harmless pair, then $(c,d)$ will be in the intersection of
$V_R$ and $V_S$.

We now suppose that $V_R$ and $V_S$ evaluated on $\globalinst$ intersect, and
show that $Q$ must hold on $\globalinst$.
First,  suppose that the intersection has a pair $(c,d)$ with $d \neq c$. Then it is clear that
$(c,d)$ must be a match of $Q$.
The more interesting case is when there is
 an element in the intersection of $V_R$ and $V_S$ of the form $(c,c)$.  

As a first subcase, suppose $(c,c)$ holds in  $S$ within
$\globalinst$  Then there are no $R$-harmless or $T$-harmless pairs of the form $(c,d)$.
Thus $(c,c)$ could not have been produced in either $V_R$ or $V_S$ as a modification, and hence
 $(c,c)$ must have gotten into view $V_R$ because $\globalinst \models S(c,c) \wedge R(c,c)$, while
$(c,c)$ was in  view $V_S$  because $\globalinst \models S(c,c) \wedge T(c,c)$. Thus we have a match
of $Q$.
The second subcase is where   $(c,c)$ does not hold in  $S$ within
$\globalinst$. Then $(c,c)$ must have been produced as a modification
of an $R$-harmless pair $(c,d)$ and as a modification of some
 $T$-harmless pair $(c,d')$. But  
from the uniqueness condition in $R$-harmlessness and $T$-harmlessness,
we conclude that $d=d'$.  Now $(c,d)$  is a match of $Q$.

\medskip

We now argue  that for this example
there are no DCQ views that are useful for $Q$ and UN non-disclosing for $\secret$.

We will prove something more general. 
Instance $\inst_1$ is a \emph{subinstance} of $\inst_2$ if $\inst_1$ is  a subset
of $\inst_2$ when they are seen as sets of facts. In other words, for every relation, its interpretation
 in $\inst_1$ is a subst of its interpretation in $\inst_2$.
A set of views $\views$ is \emph{monotone}
if whenever we have $\inst_1$ and $\inst_2$ with $\inst_1$ is a subinstance
of $\inst_2$, then for each view $\view \in \views$, the output of $\view$
on $\inst_1$ is a subset of the output on $\inst_2$.
Note that UCQs, DCQs, along with their extensions with inequalities,
are all monotone.  Datalog queries are  also monotone.

\begin{proposition} \label{prop:monotone}
 In the example, there are no monotone views that are useful for $Q$ and UN non-disclosing for $\secret$.
\end{proposition}
\begin{proof}
Say that an element $d$  \emph{appears non-trivially} in  a $k$-ary relation $S$
if there is a tuple $\vec t$ in $S$ with  $t_i=d$ ,
 a  $\vec t'$ formed from $\vec t$ by setting $t'_i$ to $d' \neq d$ while $t_j=t_i$ for $j \neq i$,
such that $\vec t'$ is not in $S$.
 For the output of a CQ or UCQ view, there is no difference
between appearing non-trivially and being in the active domain of the view output.
But for unsafe views there is a difference, since it is possible that every
element appears in the output, but  if $d$ appears non-trivially in the output of a DCQ
then there must be some disjunct in the DCQ such that $d$ satisfies the disjunct.

For elements $c$ and $d$ let $\inst_{c,d}$  be the instance where each relation
consists of the single pair $(c,d)$.
We first claim that on this instance each of $c,d$ needs to be appear non-trivially
in the output of some view.
We show the claim for $d$, with the claim for $c$ being symmetric.

Suppose not, and fix $e$ distinct from both $c$ and $d$. Consider $\inst_0= \inst_{c,d} \cup \inst_{c,e}$
along with $\inst_1$ the instance with only $(c,d)$ in $R$, both $(c,d)$ and $(c,e)$ in $S$ and in
$S$ but
 only $(c,e) \in  T$.
If $d$ does not occur non-trivially in the view output in $\inst_{c,d}$, the views must return
the same result on $\inst_{c,e}$ and $\inst_{c,d}$.
Now  $\inst_1$  is a subinstance of $\inst_{c,d}$ on the $R$ source, and  a subinstance
of $\inst_{c,e}$ on the $S$ source. Thus by monotonicity
of the views, the output of each view on $\inst_1$ is contained
in the corresponding output on $\inst_0$. On the other hand, since $\inst_0$ is a subinstance of
$\inst_1$, the view outputs must be identical on $\inst_1$ and $\inst_0$.
But since $Q$ holds on $\inst_0$ and not on $\inst_1$, this contradicts usefulness of the views.

Now consider an instance of
the form $\inst_{c,c}$ for an element $c$. The secret query $\secret$   holds, and $Q$ holds. 
So by UN non-disclosure there must be an instance $\inst'$ with the same view image as $\inst_{c,c}$ where $\secret$ fails
and $Q$ holds.
Since $Q$ holds, $\inst'$ must contain $\inst_{e,f}$ for some $e \neq f$.  By the assertions
above, there must be a view where $e$ appears
non-trivially in its output and also a view where $f$ appears non-trivially.
Clearly,
for each view output on $\inst_{c,c}$, only $c$ can appear non-trivially.
Since one of $e,f$ must be distinct from $c$, this is a contradiction of the fact
that $\inst_{c,c}$ and $\inst'$ must agree on the views.
\end{proof}

%% file: app-replication-product.tex
\subsection*{The  power of replication: proof of Proposition \ref{prop:fullrep}}
Recall the statement:

\medskip

If BCQ $\util$ contains a relation of non-zero arity replicated across all sources
then there is a d-view that is useful for $\util$ and
UN non-disclosing for BCQ $\secret$ if and only if there is no homomorphism of $\secret$ to $\util$. Further the same d-view works for all such $\secret$ for a given $\util$.

\medskip

Note that the condition on $\secret$ and $\util$ can be restated as saying
that $\util$ does not logically entail $\secret$.

For notational simplicity, we keep the replication constraints
implicit, assuming that the replicated predicates are named $T$ in
each source and $\util$ refers to this ``global'' $T$.

One direction of the theorem is clear: if there is a homomorphism of $\secret$ to $\util$
and $\views$ are useful for $\util$, then $\views$ can not be UN non-disclosing for $\secret$,
since on any instance where $\util$ holds, the views will disclose $\secret$.

For the other direction, we show, as in the case without constraints,
that there is a single d-view
that  works for any $\secret$ such that there is no homomorphism
from $\secret$ to $\util$.  We
provide views that are not isomorphism-invariant, assuming that the
active domain of instances is $Pair(\N)$ defined below.

$Pair(\N)$ is the set that contains $\N$ and is closed under pairing:
when $x$ and $y$ are allowed, then so is
$(x,y)$, the ordered pair consisting of $x$ and $y$. Note that all elements in $Pair(\N)$
have a finite $\ph$, where $\ph$ is defined as follows:
$\ph(x)=0$ for $x$ an integer of $\N$, and
$\ph((x,y))=max(\ph(x),\ph(y))+1$.

Given two instances $\inst_1$ and $\inst_2$ for the same schema,
the synchronous product of $\inst_1$ and $\inst_2$ is the instance defined
as follows:
\begin{itemize}
\item elements of the instance are pairs $(x,y)$ with $x \in \inst_1, y \in \inst_2$.
\item  for each relation $R$, we have $R((x_1,y_1) \ldots (x_n,y_n) )$ holds exactly when
$R(x_1, \ldots x_n)$ holds in $\inst_1$ and $R(y_1, \ldots y_n)$ holds in $\inst_2$
\end{itemize}

Note that the projection on the first component is a  homomorphism of the product to 
instance $\inst_1$ and projection on the second component is a homomorphism to $\inst_2$.


By a \emph{position} we mean a relation $S$ and a number between $1$ and
the arity of $S$. That is, a position describes an argument of a relation.
We assume that each variable of $\util$ is associated with a unique integer index.

We consider the  transformation  $\stretchf$ on a d-instance that
maps each local instance to its product with $\canondb(\util)$.

Note that $\stretchf(\globalinst)$ is over the same schema as
$\globalinst$ and that $\stretchf(\globalinst)$ validates the
replication constraint. Furthermore we suppose that the domain of
$\canondb(\util)$ is included in $\N$ which ensures that the minimal
$\ph$ of elements in the relation $T$ of $\stretchf(\globalinst)$ will
be the minimal $\ph$ of elements in the relation $T$ of $\globalinst$
plus one.

\begin{definition}
  We define $\equiv_s$ on $s$-instances as the reflexive transitive
  closure of $\mathcal{R}$ where $\mathcal{R}$ is defined as $\inst
  ~\mathcal{R}~ \inst'$ when $\inst=\stretchf(\inst')$ for $\inst, \inst'$ two
  $s$-instances.
\end{definition}

\begin{definition}
  We define $\equiv_G$ on d-instances as the reflexive transitive
  closure of $\mathcal{R}$ where $\mathcal{R}$ is defined as 
$\globalinst ~\mathcal{R}~ \globalinst'$ when 
$\globalinst'=\stretchf(\globalinst)$ for 
$\globalinst, \globalinst'$ two
  d-instances.
\end{definition}

\begin{definition}
  The ECR $\equiv$ is defined as $\globalinst \equiv \globalinst'
  \Leftrightarrow \bigwedge_{s \in \sources} \globalinst_s \equiv_s
  \globalinst'_S$.
\end{definition}

\begin{proposition}
  \label{prop:qi}
  For two d-instances $\globalinst$ and $\globalinst'$, when $\globalinst\vDash \util$ then
  $\globalinst \equiv \globalinst'$ if and only if $\globalinst \equiv_G \globalinst'$.
\end{proposition}

\begin{proof}
  Clearly $\globalinst \equiv_G \globalinst'$ implies $\globalinst \equiv \globalinst'$. Now let us
  suppose that $\globalinst \equiv \globalinst'$ with $\globalinst \vDash \util$ and let us
  show that $\globalinst \equiv_G \globalinst'$.

  Since $\globalinst \equiv \globalinst'$, we have, for each source $s$,  there
  exists $i_s$ such that $\globalinst_s = \stretchf^{i_s}(\globalinst'_s)$ or $\globalinst'_s =
  \stretchf^{i_s}(\globalinst_s)$, where the notation $\stretchf^i$ means iterating $\stretchf$
  $i$ times. On a d-instance $\globalinst$ where the replicated
  relation $T$ is not empty, the minimal $\ph$ of elements appearing in position
  $T[1]$ 
needs to be equal
  for all sources $s$. Therefore we cannot have distinct sources $s$
and $s'$ where $i_s \neq i_{s'}$.
And we cannot have $s$ and $s'$ such that
  $\globalinst_s = \stretchf^{i_s}(\globalinst'_s)$ and 
$\globalinst_{s'} =\stretchf^{i_{s'}} (\globalinst_{s'})$ with
  $i_s, i_{s'}>0$. Therefore, when the replicated predicate is not empty, $\globalinst
  \equiv \globalinst'$ implies $\globalinst \equiv_G\globalinst'$.
  
  When $\globalinst \vDash \util$ we have that the replicated predicate $T$ is not
  empty (since it appears in $\util$) which proves that $\globalinst \equiv_G\globalinst'$.
\end{proof}

We now show that $\equiv$ is the view that we want. 

\begin{proposition}
  The view corresponding to ECR $\equiv$ is useful for $\util$.
\end{proposition}

\begin{proof}
Let $\globalinst$ and $\globalinst'$ be two d-instances and let us suppose
$\globalinst \equiv \globalinst'$ and $\globalinst \vDash \util$.

By Proposition~\ref{prop:qi}, we have that $\globalinst \equiv_G \globalinst'$. Recall
that there is a homomorphism from $\stretchf(\globalinst)$ to $\globalinst$.
Thus it is clear that if $\stretchf(\globalinst) \models \util$ it must be
that $\globalinst \models \util$. On the other hand if we have
any match $h$ of $\util$ in $\globalinst$, we can extend it to a match
of $\util$ in the product by taking any variable $x$ of $\util$ to $(h(x), c_x)$ where
$c_x$ is the constant corresponding to $x$ in $\canondb(\util)$.
From
$\globalinst  \equiv_G \globalinst'$ either $\globalinst$ or $\globalinst'$ can be obtained from the
other by applying $\stretchf$, and so $\globalinst  \models \util$ implies
$\globalinst' \models \util$.
\end{proof}

\begin{proposition}
  The  view corresponding to ECR $\equiv$ is UN non-disclosing for $\secret$.
\end{proposition}

\begin{proof}
  Given an instance $\globalinst$, we know  that $\stretchf(\globalinst)$ has a homomorphism
into $\canondb(\util)$. 
Therefore there is no homomorphism from $\secret$ into $\stretchf(\globalinst)$
because if there were, we would have a homomorphism from $\secret$ into $\canondb(\util)$, a contradiction
of the assumption that $\util$ did not entail $\secret$.
\end{proof}

Putting together the results above we complete the proof of Proposition \ref{prop:fullrep}.

%% file: app-norelalg.tex
\subsection*{Proof of Proposition \ref{prop:norelalg}}
Recall the statement:

\medskip

There is a d-schema with a replication constraint, along with
Boolean CQs $\util$ and $\secret$ such that there is a d-view which is useful
for $\util$ and UN non-disclosing for $\secret$, but there
is no d-view whose view definitions return only facts containing values in the active domain and which commute with isomorphisms.

\medskip

For  a query $Q_V$ that always returns tuples containing only values in  the active domain of the input,
we will use the following restricted version of the isomorphism-invariance property:

\medskip

If $\inst$ and $\inst'$ are source instances for $Q_V$
 with $\inst'$ formed from $\inst$ via an isomorphism that is the identity on $\adom(V(I))$, then
$\inst$ and $h(\inst)$ agree on $Q_V$. 

\medskip

For example, any query in relational algebra
has this property.

Consider the following Boolean CQs (existentially quantifiers omitted):
\[
\util= P(x) \wedge P(y) \wedge S(y) \wedge S(z) \wedge T(z) \wedge T(x) \wedge R(w)
\]

\[
\secret= S(x) \wedge T(x) \wedge P(x)
\]

It is easy to see  that $\util$ does not entail $\secret$.
Because there is a replicated relation mentioned in the query, Proposition
\ref{prop:fullrep}
implies that we can get a d-view that is useful for $\util$
and UN non-disclosing for $\secret$.

Now, by way of contradiction,  fix a d-view $\views$ that is useful for $\util$, returns
only facts whose values lie in the active domain, and satisfies the isomorphism-invariance property.

\begin{proposition} \label{prop:agreelocal} Suppose d-view $\views$
is useful for $\util$. Given two instances $\inst_1$ and $\inst_2$ for the source
whose unreplicated relation is $G$, if $\inst_1$ and $\inst_2$ agree on the replicated
relation and agree on each view in $\views$ for this source, then they must agree on $G$.
\end{proposition}
\begin{proof} 
Let $c \in G(\inst_1)$ and form $\globalinst_1$ by choosing a  context for the other unshared relations
with each relation containing only $c$.  Let $\globalinst_2$ be formed similarly from $\inst_2$.
Then $\globinst_1$ and $\globalinst_2$ agree on all views in $\views$.
There is a match of $\utility$ on $\globalinst_1$ so the same must be true of $\globalinst_2$, since $\views$ is useful
for $\utility$. Clearly the witness
can only be $c$, thus $c \in G(\inst_2)$.
\end{proof}

By way of contradiction, 
we  fix a d-view $\views$ that is useful for $\util$
and UN non-disclosing for $\secret$, where each view in $\views$ returns facts containing only elements of the active
domain and satisfies the isomorphism-invariance property.

Let $\globalinst_0$ be the critical instance, recalling that this
has only a single element $\critelement$ with every relation holding of it.
We show that
$\globalinst_0$ contradicts that $\views$ is UN non-disclosing for $\secret$.
Note that $\adom(\views_0) \subseteq \{\critelement \}$ since we assume views return facts containing only elements in the active domain
of the input instance.

\begin{proposition}  \label{prop:containedshared} Let $\globalinst$  agree with $\globalinst_0$ on $\views$.
For any unreplicated relation $G$, 
we have $\globalinst \models \forall x. ( (G(x) \wedge x \neq \critelement) \rightarrow R(x) )$.
\end{proposition}
\begin{proof}
Let $\inst$ be the restriction of $\globalinst$ to the
$G$ source, and suppose $G(v)$ holds in $\inst$ for $v \neq \critelement$. Let $c$ be a fresh value and $\inst'$ be 
the result of
applying an isomorphism  swapping $v$ and $c$. Since $v \neq \critelement$ the isomorphism-invariance property
implies that $\inst$ and $\inst'$ agree on the views. Since $\inst'$ and $\inst$  disagree on $G$,  Proposition \ref{prop:agreelocal}
implies they cannot agree on the replicated relation. Note that $c$ was fresh  (hence in particular
was not in the interpretation of $R$ within $\inst$),
and no other change occurred in $R$
outside of the swap of $v$ and $c$, the only way that $\inst$ and $\inst'$ can disagree
on  $R$ is because $R(v)$ held in $\inst$.
\end{proof}

By usefulness of $\views$ and the fact that $\globalinst_0$ has
a match of $\util$, we know that $\globalinst$ has a match of $\util$ with $x_0,y_0,z_0, w_0$.
We first consider the case where  two of $x_0, y_0, z_0$ are the same.
In this case $\secret$ has a match.

Now suppose the match has all of  $x_0, y_0, z_0$ distinct.
At most one of them can be $\critelement$, so assume $y_0$ and $z_0$ are not $\critelement$.
Proposition \ref{prop:containedshared} implies that they are both in $R$.
Since $y_0$ and $z_0$ are not in $\adom(\views_0)$, when we consider the result of swapping
$y_0$ and $z_0$, it does not impact the views, by the isomorphism-invariance property.
Since this swap also does not change $R$, we can apply Proposition \ref{prop:agreelocal}
to conclude that all local sources agree on $y_0$ and $z_0$.
In particular we have $P(y_0) \wedge S(z_0)$, thus we have $P(z_0) \wedge S(y_0)$.
But then either $y_0$ or $z_0$ gives a match of $\secret$.

%% file: app-repnomincq.tex
subsection*{Proof of Theorem \ref{thm:nodisclosuremincq}}

Recall the statement:

\medskip

There is a schema with a replication constraint and a utility query $\util$
such that
 there is no CQ-based d-view $\views$ that  is minimal for UN non-disclosure
within the class of CQ views.
In particular, there is no CQ minimally informative useful d-view within the
class of CQ views.

\medskip

Recall also the example schema and query that we introduced in the body of the paper.
Our schema has two sources, $\Pcal$ and $\Scal$ all containing an eponymous
relation (respectively $P$ and $S$), both sources contain a
shared relation $T$. All relations are binary.

The utility query is $\util= \exists w,x,y,z . T(x,y) \land S(y,z) \land T(z,w) \land
P(w,x)$ or graphically:

\begin{figure}[h]
  \centering
  \begin{tikzpicture}[node distance=2cm, auto]
    \node (x) at (0,0) {$x$};
    \node (y) at (2,0) {$y$};
    \node (z) at (2,2) {$z$};
    \node (w) at (0,2) {$w$};

    \draw[->] (x) -- (y) node [midway, below] {$T$};;
    \draw[->] (z) -- (w) node [midway, above] {$T$};;
    \draw[->] (y) -- (z) node [midway, right] {$S$};;
    \draw[->] (w) -- (x) node [midway, left] {$P$};;
    
  \end{tikzpicture}
\end{figure}

Finally recall the three secret queries introduced in the body.
\begin{align*}
\secret_1= \exists x. S(x,x) \\
\secret_2= \exists x,y. T(x,y) \land S(y,x) \\
\secret_3= \exists x,y,z. T(x,y)\land S(y,z) \land T(z,x)
\end{align*}

We noted in the body of the paper that for each $i= 1 \ldots 3$ there
is a d-view that is useful for $\util$ and UN non-disclosing for
$\secret_i$:
\begin{itemize}
\item $Q_{\Scal}(x,w) = \exists y,z. T(x,y)\land S(y,z) \land T(z,w)$ and
  $Q_{\Pcal}(w,x) = P(w,x)$ for $\secret_1$,
\item $Q_{\Scal}(y,w) = \exists z.  S(y,z) \land T(z,w)$ and
$Q_{\Pcal}(x,z) = \exists y. P(x,y) \land T(y,z) $ for $\secret_2$, or
\item $Q_{\Scal}(y,w) = \exists z.  T(y,z) \land S(z,w)$ and
$Q_{\Pcal}(x,z) = \exists y. T(x,y) \land P(y,z) $
 for $\secret_2$,
\item $Q_{\Scal}(x,y) = S(x,y)$ and
$Q_{\Pcal}(x,w) = \exists y,z. T(x,y)\land P(y,z) \land T(z,w)$
 for $\secret_3$.
\end{itemize}

Now consider a d-view  $\Vcal$, useful  for $\util$, 
homomorphism-invariant, and $\adom$-based. We will show that $\Vcal$ is necessarily UN disclosing
for one of the secrets among $\secret_1 \dots \secret_3$.
We claimed in Lemma \ref{lem:oneofthese}:

\medskip

Any  d-view  that is  useful  for $\util$,
homomorphism-invariant, and $\adom$-based must  necessarily UN disclosing
for one of the secrets among $\secret_1 \dots \secret_3$.

\medskip

The remainder of this subsection is devoted to the proof of this lemma,
which completes the theorem.

Let us consider the following d-instance $\D=CanonDB(\util)$ shown below.

\begin{figure}[h]
  \centering
  \begin{tikzpicture}[node distance=2cm, auto]
    \node (a) at (0,0) {$a$};
    \node (b) at (2,0) {$b$};
    \node (c) at (2,2) {$c$};
    \node (d) at (0,2) {$d$};

    \draw[->] (a) -- (b) node [midway, below] {$S$};;
    \draw[->] (c) -- (d) node [midway, above] {$P$};;
    \draw[->] (b) -- (c) node [midway, right] {$T$};;
    \draw[->] (d) -- (a) node [midway, left] {$T$};;
    
  \end{tikzpicture}
  \caption{Instance $\D$}
\end{figure}

\begin{claim}
  \label{c1}
  $\adom(\Vcal[\Scal](\D))$ includes either $a$ or $d$. 
\end{claim}

\begin{proof}

  Let us consider the instance $\I_1$ (depicted below) and let us show
  that 
\begin{enumerate}
\item  $\adom(\Vcal[\Scal](\I_1))$ includes $a$ or $d$ if
  and only if $\adom(\Vcal[\Scal](\D))$ also does and 
\item
  $\adom(\Vcal[\Scal](\I_1))$ includes either $a$ or $d$.
\end{enumerate}

\begin{figure}[h]
  \centering
  \begin{tikzpicture}[node distance=2cm, auto]
    \node (a) at (0,0) {$a$};
    \node (b) at (2,0) {$b$};
    \node (c) at (2,2) {$c$};
    \node (d) at (0,2) {$d$};
    \node (e) at (-2,0) {$e$};
    \node (f) at (-2,2) {$f$};
    \node (g) at (4,0) {$g$};
    \node (h) at (4,2) {$h$};

    \draw[->] (f) -- (e) node [midway, left] {$T$};;
    \draw[->] (g) -- (h) node [midway, right] {$T$};;
    \draw[->] (a) -- (b) node [midway, below] {$S$};;
    \draw[->] (c) -- (d) node [midway, above] {$P$};;
    \draw[->] (b) -- (c) node [midway, right] {$T$};;
    \draw[->] (d) -- (a) node [midway, left] {$T$};;
    
  \end{tikzpicture}
  \caption{Instance $\I_1$}
\end{figure}

For the first item since $\D\subseteq \I_1$ we know that
$\adom(\Vcal[\Scal](\D))\subseteq \adom(\Vcal[\Scal](\I_1))$. 
 Consider the homomorphism $\mu$ sending:
$f$ to $d$, $e$ to $a$, $h$ to $c$, and $g$ to $b$.
The elements $d$ and $a$ are mapped to themselves.
$g$), to $d$ (resp. $a$, $c$ and $b$). 
Since $\mu(\I_1)=\D$ we know from homomorphism-invariance that
 $\mu(\Vcal[\Scal](\I_1))\subseteq \Vcal[\Scal](\D)$. If
$\adom(\mu(\Vcal[\Scal](\I_1)))$ contains $a$ or $d$ 
we will have that $\adom(\mu(\Vcal[\Scal](\D)))$
contains $a$ or $d$. Note that this also proves that we cannot have
$e$ or $f$ in $\adom(\mu(\Vcal[\Scal](\I_1)))$ if we do not also have 
$a$ or $d$ in $\adom(\mu(\Vcal[\Scal](\I_1)))$.

To prove that $\adom(\Vcal[\Scal](\I_1))$ includes $a$ or $d$, let us
consider the instance $\I_2$ (depicted below) which is the
instance $\I_1$ where the $\Scal$ source has been replaced by its
image through the isomorphism $\nu$ that exchanges $e$ and $a$ and $f$
and $d$.

\begin{figure}[h]
  \centering
  \begin{tikzpicture}[node distance=2cm, auto]
    \node (a) at (0,0) {$a$};
    \node (b) at (2,0) {$b$};
    \node (c) at (2,2) {$c$};
    \node (d) at (0,2) {$d$};
    \node (e) at (-2,0) {$e$};
    \node (f) at (-2,2) {$f$};
    \node (g) at (4,0) {$g$};
    \node (h) at (4,2) {$h$};

    \draw[->] (f) -- (e) node [midway, left] {$T$};
    \draw[->] (g) -- (h) node [midway, right] {$T$};
    \draw[->] (e) to [bend right] node [midway, below] {$S$} (b);
    \draw[->] (c) -- (d)  node [midway, above] {$P$};
    \draw[->] (b) -- (c) node [midway, right] {$T$};
    \draw[->] (d) -- (a) node [midway, left] {$T$};
    
  \end{tikzpicture}
  \caption{Instance $\I_2$}
\end{figure}

If we suppose that $\{a,d\}\cap \adom(\Vcal[\Scal](\I_1))=\emptyset$
then we also have $\{f,e\}\cap \adom(\Vcal[\Scal](\I_1))=\emptyset$.
When $\{a,d,e,f\}\cap \adom(\Vcal[\Scal](\I_1))=\emptyset$ we have
that $\Vcal[\Scal](\I_1)=\nu(\Vcal[\Scal](\I_1))$. Now, since $\nu$ is
an isomorphism between $\I_1[\Scal]$ and $\I_2[\Scal]$ we have
$\Vcal[\Scal](\I_1)=\nu(\Vcal[\Scal](\I_1))=\Vcal[\Scal](\nu(\I_1))=\Vcal[\Scal](\I_2)$
and as the $\Pcal$ sources are identical we conclude that
$\Vcal(\I_1)=\Vcal(\I_2)$ which would mean that $\Vcal$ is not useful
($Q_U\vDash \I_1$ but $Q_U\not\vDash \I_2$).  
\end{proof}

\begin{claim}
  \label{c2}
  $\adom(\Vcal[\Scal](\D))$ contains either $b$ or $c$. 
\end{claim}

\begin{proof}
  Same proof as above (by symmetry).
\end{proof}


\begin{claim}
  \label{c3}
  We have that $\Vcal$ is UN-disclosing:
  \begin{enumerate}
    \item For $\secret_1$ when $\{a,b\}\subseteq adom(\Vcal[\Scal](\D))$.
    \item For $\secret_2$ when $\{d,b\}\subseteq adom(\Vcal[\Scal](\D))$.
    \item For $\secret_2$ when $\{a,c\}\subseteq adom(\Vcal[\Scal](\D))$.
    \item For $\secret_3$ when $\{d,c\}\subseteq adom(\Vcal[\Scal](\D))$.
  \end{enumerate}
\end{claim}

\begin{proof}
The four cases are symmetric so we will focus on the second one.

Let us consider the critical d-instance, $\I= \{T(u,u), S(u,u),P(u,u)\}$, and let
us take a d-instance $\E$ such that $\Vcal(\E)=\Vcal(\I)$. Because
$\Vcal$ is useful we know that there exists $x,y,z,w$ such that
$T(x,y) \land S(y,z) \land T(z,w) \land P(w,x)$ holds in $\E$ (the
$x$, $y$, $z$ and $w$ being not necessarily different).

Now we can consider the homomorphism $\xi$ that maps $d$ to $x$, $a$
to $y$, $b$ to $z$ and $c$ to $w$. $\xi$ maps $\D$ to a subset of $\E$
and this proves that $\xi(\Vcal[S](\D))\subseteq \Vcal[S](\E) =
\Vcal[S](\I)$. But $\adom(\Vcal[S](\I))=\{u\}$ and we know that
$\xi(d)=x$ and $\xi(b)=z$ belongs to $\adom(\Vcal[S](\I))$ which means
that $x=u=z$ and thus $\E \models \secret_2$. Since $\E$ was an arbitrary
instance agreeing with the views we have shown that $\Vcal$ is UN-disclosing.
\end{proof}

By combining claims~\ref{c1}, \ref{c2}, and \ref{c3}, we get
the proof of Lemma \ref{lem:oneofthese}.

%% file: despostfull.bbl
\begin{thebibliography}{}

\bibitem[\protect\citeauthoryear{Abiteboul, Hull, and Vianu}{1995}]{AHV}
Abiteboul, S.; Hull, R.; and Vianu, V.
\newblock 1995.
\newblock {\em Foundations of Databases}.
\newblock Addison-Wesley.

\bibitem[\protect\citeauthoryear{Arenas, Barcel{\'{o}}, and
  Reutter}{2011}]{arenasucq}
Arenas, M.; Barcel{\'{o}}, P.; and Reutter, J.~L.
\newblock 2011.
\newblock Query languages for data exchange: Beyond unions of conjunctive
  queries.
\newblock {\em Theory Comput. Syst.} 49(2):489--564.

\bibitem[\protect\citeauthoryear{Baget \bgroup et al\mbox.\egroup
  }{2011}]{frontier1}
Baget, J.-F.; Mugnier, M.-L.; Rudolph, S.; and Thomazo, M.
\newblock 2011.
\newblock Walking the complexity lines for generalized guarded existential
  rules.
\newblock In {\em IJCAI}.

\bibitem[\protect\citeauthoryear{Baget \bgroup et al\mbox.\egroup
  }{2014}]{moreacyclicity}
Baget, J.; Garreau, F.; Mugnier, M.; and Rocher, S.
\newblock 2014.
\newblock Extending acyclicity notions for existential rules.
\newblock In {\em {ECAI}}.

\bibitem[\protect\citeauthoryear{Bater \bgroup et al\mbox.\egroup
  }{2017}]{jenniesecurequery}
Bater, J.; Elliott, G.; Eggen, C.; Goel, S.; Kho, A.~N.; and Rogers, J.
\newblock 2017.
\newblock {SMCQL:} secure query processing for private data networks.
\newblock In {\em {VLDB}}.

\bibitem[\protect\citeauthoryear{Benedikt \bgroup et al\mbox.\egroup
  }{2016}]{lics16}
Benedikt, M.; Bourhis, P.; ten Cate, B.; and Puppis, G.
\newblock 2016.
\newblock Querying visible and invisible information.
\newblock In {\em LICS}.

\bibitem[\protect\citeauthoryear{Benedikt \bgroup et al\mbox.\egroup
  }{2019}]{ijcai19}
Benedikt, M.; Bourhis, P.; Jachiet, L.; and Thomazo, M.
\newblock 2019.
\newblock Reasoning about disclosure in data integration in the presence of
  source constraints.
\newblock In {\em IJCAI}.

\bibitem[\protect\citeauthoryear{Benedikt, {Cuenca Grau}, and
  Kostylev}{2018}]{privacyjair}
Benedikt, M.; {Cuenca Grau}, B.; and Kostylev, E.~V.
\newblock 2018.
\newblock Logical foundations of information disclosure in ontology-based data
  integration.
\newblock {\em Artif. Intell.} 262:52--95.

\bibitem[\protect\citeauthoryear{Benedikt, ten Cate, and
  Tsamoura}{2016}]{ustods}
Benedikt, M.; ten Cate, B.; and Tsamoura, E.
\newblock 2016.
\newblock Generating plans from proofs.
\newblock In {\em TODS}.

\bibitem[\protect\citeauthoryear{Bonatti and
  Sauro}{2013}]{DBLP:conf/semweb/BonattiS13}
Bonatti, P.~A., and Sauro, L.
\newblock 2013.
\newblock A confidentiality model for ontologies.
\newblock In {\em ISWC}.

\bibitem[\protect\citeauthoryear{Calvanese \bgroup et al\mbox.\egroup
  }{2012}]{DBLP:journals/jcss/CalvaneseGLR12}
Calvanese, D.; {De Giacomo}, G.; Lenzerini, M.; and Rosati, R.
\newblock 2012.
\newblock {View-based Query Answering in Description Logics: Semantics and
  Complexity}.
\newblock {\em J. Comput. Syst. Sci.} 78(1):26--46.

\bibitem[\protect\citeauthoryear{Chaum, Cr{\'e}peau, and
  Damgard}{1988}]{infotheoreticprivccd}
Chaum, D.; Cr{\'e}peau, C.; and Damgard, I.
\newblock 1988.
\newblock Multiparty unconditionally secure protocols.
\newblock In {\em STOC}.

\bibitem[\protect\citeauthoryear{{Cuenca Grau} \bgroup et al\mbox.\egroup
  }{2013}]{acyclicity}
{Cuenca Grau}, B.; Horrocks, I.; Kr{\"{o}}tzsch, M.; Kupke, C.; Magka, D.;
  Motik, B.; and Wang, Z.
\newblock 2013.
\newblock Acyclicity notions for existential rules and their application to
  query answering in ontologies.
\newblock {\em JAIR} 47:741--808.

\bibitem[\protect\citeauthoryear{Deutsch, Nash, and
  Remmel}{2008}]{chaserevisited}
Deutsch, A.; Nash, A.; and Remmel, J.
\newblock 2008.
\newblock The chase revisited.
\newblock In {\em PODS}.

\bibitem[\protect\citeauthoryear{Dwork and Roth}{2014}]{diffprivj}
Dwork, C., and Roth, A.
\newblock 2014.
\newblock The algorithmic foundations of differential privacy.
\newblock {\em Found. \& Trends in Th. Comp. Sci.} 9(3\&4):211--407.

\bibitem[\protect\citeauthoryear{Dwork}{2006}]{diffpriv}
Dwork, C.
\newblock 2006.
\newblock Differential privacy.
\newblock In {\em ICALP}.

\bibitem[\protect\citeauthoryear{Fagin \bgroup et al\mbox.\egroup
  }{2005}]{fagindataex}
Fagin, R.; Kolaitis, P.~G.; Miller, R.~J.; and Popa, L.
\newblock 2005.
\newblock Data exchange: {S}emantics and query answering.
\newblock {\em Theoretical Computer Science} 336(1):89--124.

\bibitem[\protect\citeauthoryear{Gogacz and Marcinkowski}{2015}]{redspider}
Gogacz, T., and Marcinkowski, J.
\newblock 2015.
\newblock The hunt for a red spider: Conjunctive query determinacy is
  undecidable.
\newblock In {\em LICS}.

\bibitem[\protect\citeauthoryear{Halevy}{2001}]{DBLP:journals/vldb/Halevy01}
Halevy, A.~Y.
\newblock 2001.
\newblock Answering queries using views: {A} survey.
\newblock {\em {VLDB} J.} 10(4):270--294.

\bibitem[\protect\citeauthoryear{Li \bgroup et al\mbox.\egroup
  }{2017}]{pricingprivate}
Li, C.; Li, D.~Y.; Miklau, G.; and Suciu, D.
\newblock 2017.
\newblock A theory of pricing private data.
\newblock {\em Commun. {ACM}} 60(12):79--86.

\bibitem[\protect\citeauthoryear{Maier, Mendelzon, and Sagiv}{1979}]{chaseorig}
Maier, D.; Mendelzon, A.~O.; and Sagiv, Y.
\newblock 1979.
\newblock {Testing implications of data dependencies}.
\newblock {\em TODS} 4(4):455--469.

\bibitem[\protect\citeauthoryear{Marnette}{2009}]{marnettecritinst}
Marnette, B.
\newblock 2009.
\newblock {Generalized schema-mappings: from termination to tractability}.
\newblock In {\em PODS}.

\bibitem[\protect\citeauthoryear{Nash and Deutsch}{2007}]{dnprivacy}
Nash, A., and Deutsch, A.
\newblock 2007.
\newblock Privacy in {GLAV} information integration.
\newblock In {\em ICDT}.

\bibitem[\protect\citeauthoryear{Nash, Segoufin, and Vianu}{2010}]{NSV}
Nash, A.; Segoufin, L.; and Vianu, V.
\newblock 2010.
\newblock Views and queries: Determinacy and rewriting.
\newblock {\em TODS} 35(3).

\bibitem[\protect\citeauthoryear{Onet}{2013}]{onet}
Onet, A.
\newblock 2013.
\newblock The chase procedure and its applications in data exchange.
\newblock In {\em DEIS},  1--37.

\end{thebibliography}
